\documentclass{siamart251216}
\usepackage{braket}
\usepackage{amsmath,amssymb,amsfonts}
\usepackage{algorithmic}
\usepackage{graphicx}
\usepackage{textcomp}

\usepackage{color} 
 \usepackage{amsfonts}
 \usepackage{amsmath}
 \usepackage{amssymb}
 \usepackage{array}
 \usepackage{mathrsfs}
 \usepackage{leftidx}
 \usepackage{graphicx}
 \usepackage{epstopdf}
 \usepackage{dcolumn}
 \usepackage{bm}

\usepackage{cite}

\usepackage{multirow}

\usepackage{cases}

\usepackage{lipsum}
\usepackage{amsfonts}
\usepackage{graphicx}
\usepackage{epstopdf}
\usepackage{algorithmic}
\usepackage{lineno}
\ifpdf
  \DeclareGraphicsExtensions{.eps,.pdf,.png,.jpg}
\else
  \DeclareGraphicsExtensions{.eps}
\fi

\newsiamremark{remark}{Remark}
\newsiamremark{hypothesis}{Hypothesis}
\crefname{hypothesis}{Hypothesis}{Hypotheses}
\newsiamthm{claim}{Claim}

\newtheorem{assumption}{Assumption}

\def\BibTeX{{\rm B\kern-.05em{\sc i\kern-.025em b}\kern-.08em
    T\kern-.1667em\lower.7ex\hbox{E}\kern-.125emX}}

\title{A Lie-algebraic approach to non-Markovian quantum dynamics}

\author{Haijin Ding\thanks{Department of Mathematics, The University of Hong Kong, Pokfulam
Road, Hong Kong, Hong Kong SAR, P.R. China. (\email{dhj17@tsinghua.org.cn}).} \and Stephen S.-T. Yau\thanks{Corresponding author. Department of Mathematical Sciences, Tsinghua
University, Beijing 100084, China (\email{yau@uic.edu}).} \and Zhiwen Zhang\thanks{Corresponding author. Department of Mathematics, The University of Hong Kong, Pokfulam
Road, Hong Kong, Hong Kong SAR, P.R. China, and Materials Innovation Institute for Life Sciences and Energy (MILES),
HKU-SIRI, Shenzhen, 518045, P.R. China (\email{zhangzw@hku.hk}).}}

\headers{A Lie algebraic approach to non-Markovian quantum dynamics}{Haijin Ding,  Stephen S.-T. Yau and Zhiwen Zhang}

\begin{document}

\maketitle

\begin{abstract}
In this paper, we study the non-Markovian quantum dynamics in quantum computations from the perspective of a Lie algebraic approach based on numerical analysis. By vectorizing the density matrix of quantum states, the non-Markovian evolutions can be represented with high-dimensional linear time-varying equations, where the time-varying parameters arise from the non-Markovian interactions between the quantum system and environment. We study the Magnus expansion of such linear time-varying quantum dynamics and clarify how the truncation errors for the first- and second-order Magnus expansions are influenced by the non-Markovian properties. Besides,  when the quantum states are measured for filtering, the dynamics can be modeled as time-varying stochastic differential equations due to the existence of measurement noise. The Magnus expansions based on quantum stochastic filtering are different when the quantum measurement noises are modeled in an {It\^{o}} or Stratonovich approach, rendering different truncation errors. Based on this, numerical simulations further demonstrate the efficiency of Magnus expansions in simulating non-Markovian quantum dynamics without or with stochasticity, and how the truncation errors are influenced by the Lie algebras in the Liouville space. 
\end{abstract}



\section{Introduction}
Quantum computation has attracted much attention for its potential advantages in solving mathematical problems that are difficult for classical computation methods~\cite{nielsen2010quantum,divincenzo1995quantum,simon1997power}, i.e., solving problems that are $\rm NP$ hard~\cite{kaminsky2004scalable}, simulating complex chemical or physical systems~\cite{benioff1980computer,feynman2018simulating,schleich2026chemically} and accelerating algorithms in artificial intelligence~\cite{cong2019quantum,zhu2023artificial}. Currently, quantum computation realizations are based on various platforms such as superconducting circuits~\cite{arute2019quantum,wu2021strong}, cold atoms~\cite{ma2023high}, molecules~\cite{bao2023dipolar}, photons~\cite{zhong2020quantum}, and so on. Among different approaches, the core implementations of quantum computational algorithms are based on the control and measurement of quantum systems. 

The evolution of quantum states in a closed system independent of the environment is governed by the Schr\"{o}dinger equation~\cite{boussaid2013weakly}, and the dynamics of the quantum state vector can be regarded as a bilinear system when time-varying control fields are applied to the quantum system~\cite{mirrahimi2005lyapunov,duca2024small}. Due to the unitary property of the Hamiltonian, the evolution of quantum states in the closed system can be regarded as rotations on unitary groups. The numerical analysis of quantum dynamics based on Schr\"{o}dinger equation has been studied from the perspective of Magnus expansion~\cite{SIAM2003magnus,iserles2018magnus} and Newton method~\cite{von2010globalized}. Furthermore, the numerical frameworks can be extended to simulations of various Hamiltonian and gate operations in quantum computations~\cite{casares2024quantum}.

However, practical quantum systems always interact with the environment via the decoherence process. As a result, an excited state with higher energy levels can decay to its ground state with lower energy, and pure quantum states can become mixed after interacting with the environment~\cite{breuer2002theory}. On the one hand, the evolution of quantum states in the open system can be modeled as stochastic Schr\"{o}dinger equations with different trajectories~\cite{plenio1998quantum}. On the other hand, a quantum state in the density matrix format is governed by the Lindblad master equation, the Lindbladian representing its interaction with the environment~\cite{breuer2002theory}. For the Markovian quantum master equation without memory effects, the density matrix can be vectorized to the high-dimensional column vector that evolves in the Liouville space~\cite{markowich1989equivalence}. Additionally, for the more general circumstance that quantum systems interact with the environment in a non-Markovian format~\cite{diosi1997nonPLA,diosi1998non,strunz1999open}, the amplitudes of the Lindbladian will be time-varying integrals, rendering time-varying parameters in the equivalent Liouville equation~\cite{smirne2010nakajima,ding2024non}. 

The measurement, filtering, and feedback control of quantum systems are pivotal to both quantum optics and quantum computing~\cite{wiseman2009quantum,bouten2007introduction,belavkin1992quantum,ding2026nonAuto}. For example, feedback control using the measurement information of quantum states can correct the error bits in quantum computations~\cite{sarovar2004practical}. In addition, filtering methods can identify parameters~\cite{chase2009single} and readout quantum states~\cite{readout} in the realizations of quantum algorithms. Due to the existence of measurement noise, the quantum filtering dynamics can be modeled as a stochastic master equation and can also be generalized to the high-dimensional  Liouville equation with stochasticity. Recently, numerical methods and scientific computations on stochastic differential equations have attracted much attention~\cite{kamm2021stochastic,wang2020magnus}, while this has not been generalized to the stochastic evolutions of quantum systems. 

In this paper, we study the numerical methods for general quantum dynamics with non-Markovianity from the perspective of Magnus expansion and Lie algebras. The main contributions of this research work are summarized as follows:
\begin{itemize}
\item In Theorem~\ref{FirstOrderError} and Theorem~\ref{SecondOrderError}, we rigorously clarify how the truncation errors of the first- and second-order Magnus expansions are influenced by the non-Markovian interaction between the quantum system and environment.
\item In Theorem~\ref{convergence}, we illustrate the condition on the convergence and improvements of Magnus expansions in simulating non-Markovian quantum dynamics in a rotational interaction picture. 
\item In Theorem~\ref{Itoerror}, Theorem~\ref{ItoStoMag} and Theorem~\ref{StrotovichStoMag}, we compare the difference between Magnus expansions based on {\rm{It\^{o}}} and Stratonovich stochastic dynamics in quantum filtering.
\end{itemize}

The remainder of this paper is organized as follows. Section~\ref{Sec:Model} introduces the modeling of non-Markovian quantum dynamics in the Liouville space and the Magnus expansion for the circumstance without measurements. Section~\ref{Sec:Stochastic} studies the Magnus expansion for non-Markovian quantum stochastic dynamics when the quantum system is measured, especially on the difference between the Magnus expansions for {\rm{It\^{o}}} and Stratonovich stochastic dynamics. Numerical simulations for the above non-Markovian quantum dynamics are presented in~\ref{Sec:numerical}. Section~\ref{Sec:Conclusion} concludes this paper.  
\section{Non-Markovian quantum stochastic dynamics} \label{Sec:Model}
In this section, we first introduce the non-Markoivan quantum stochastic dynamics with measurements, which can be modeled as a stochastic master equation (SME) and an equivalent vectorized format. Based on this, we study the non-Markovian dynamics from the perspective of Lie algebras.  
\subsection{Modeling with stochastic master equations}
Quantum systems interacting with a non-Markovian environment can be modeled as the following stochastic Schr\"{o}dinger equation~\cite{diosi1998non,diosi1997nonPLA}
\begin{equation} \label{con:SSENonMar}
\begin{aligned}
\frac{d}{dt} |\psi(t)\rangle =& -iH |\psi(t)\rangle  + \left[L |\psi(t)\rangle z_t -L^{\dag}\int_0^t \alpha(t,s)\frac{\delta |\psi(t)\rangle}{\delta z_s}ds\right],
\end{aligned}
\end{equation}
where $|\psi(t)\rangle$ represents the state vector of a quantum system with Hamiltonian $H$, $L$ represents the interaction operator between the quantum system and the environment, $L^{\dag}$ is the  Hermitian conjugate of the operator $L$, the time-varying stochastic variable $z_t$ represents the influence of the environment on the quantum system, the integral kernel $\alpha(t,s)$ represents the non-Markovian property of the system-environment interaction and
\begin{equation} \label{con:NonMarkovAlpha}
\alpha(t,s) =E\left[z_tz_s^*\right] =  \frac{\gamma}{2} e^{-\gamma|t-s|-i\Omega(t-s)},
\end{equation}
where $\gamma^{-1}$ represents the environmental memory time scale, $\Omega$ represents the central frequency for the modeled oscillators, $E[\bullet]$ represents the average of a stochastic process, and $*$ is for the complex conjugate. Obviously, 
\begin{equation} \label{con:MarkovAlpha}
\lim_{\gamma \rightarrow \infty} \alpha(t,s)  = \delta(t-s),
\end{equation}
and the  Schr\"{o}dinger equation (\ref{con:SSENonMar}) reduces to be Markovian when $\gamma \rightarrow \infty$. 

Apart from the state vector dynamics in Eq.~(\ref{con:SSENonMar}), the quantum system can also be represented with the density matrix $\rho(t) = |\psi(t)\rangle \langle \psi(t)|$, and the non-Markovian dynamics in Eq.~(\ref{con:SSENonMar}) can be represented by a non-Markovian master equation after averaging $z_t$. Then the quantum dynamics can be modeled as 
\begin{equation} \label{con:MasterEq}
\begin{aligned}
\frac{d\rho}{dt}   =&-i \left[H,\rho\right]  + \left[L,\rho O^{\dag} \right] + \left[O\rho, L^{\dag}\right] ,
\end{aligned}
\end{equation}
where the commutator $[A,B] = AB-BA$, the operator 
\begin{equation}
\begin{aligned}
O =  \sum_{j =1}^N f_j(t) L_j  ,
\end{aligned}
\end{equation}
and
\begin{equation}\label{con:blackgdefine}
\begin{aligned}
f_j(t) = \int_0^t \beta_j(t,s)\alpha(t,s) ds, 
\end{aligned}
\end{equation}
where $L$ represents the dissipation operator of the quantum system to the environment, $L_j$ represents the non-Markoivan interaction between the quantum and the environment with memory effects, and $f_j(t)$ is the time-varying Lindblad amplitude determined by the integral kernel $\alpha(t,s)$ and the implicit function $\beta_j(t,s)$ ~\cite{diosi1998non,ding2024non}. The parameter settings are constrained by the following assumption.

\begin{assumption}
The time-varying amplitudes $f_j(t)$ are bounded.
\end{assumption} 

In addition, when the quantum system is coupled to a measurement apparatus via the operator $M$, we further define the Hermitian operator $\mathcal{M} = M+ M^{\dag}$ for the acquisition of real-valued measurements; then the measurement result of the quantum system via Homodyne detection reads
\begin{equation} \label{con:Ic}
\begin{aligned}
\mathcal{Y}(t)   = \left\langle \mathcal{M} \right\rangle  + \frac{d W_t}{d t},
\end{aligned}
\end{equation}
where the derivative  $\xi\left(t\right) = dW_t/dt $ is regarded as Gaussian white noise,  $d W_t$ is a Wiener increment that satisfies $E\left[d W_t\right] = 0$ and $E\left[d W_t^2\right] = dt$. Due to the existence of measurement noise, the dynamics of quantum states with continuous detection is governed by the {\rm{It\^{o}}} stochastic master equation (SME) as~\cite{wiseman1994quantum,ticozzi2012stabilization,dingSIAM} 
\begin{equation} \label{con:SME}
\begin{aligned}
d\rho   =&-i \left[H,\rho\right]dt  + \mathcal{L}_{O}\left[\rho\right]dt + \mathcal{H}[M]\rho d W_t,
\end{aligned}
\end{equation}
where $\mathcal{L}_{O}\left[\rho\right] = \left[L,\rho O^{\dag} \right] + \left[O\rho, L^{\dag}\right]$ according to Eq.~(\ref{con:MasterEq}) and the last term induced by quantum measurements reads
\begin{equation} \label{con:HM}
\begin{aligned}
\mathcal{H}[M]\rho = M\rho + \rho M^{\dag} - {\rm Tr}\left[\left(M+M^{\dag}\right)\rho \right]\rho.
\end{aligned}
\end{equation}

\begin{remark}
The stochastic master equation (\ref{con:SME}) is a combination of linear and nonlinear terms. The linear part is determined by the quantum system's free Hamiltonian and non-Markovian interactions with the environment. The nonlinear part is determined by the measurement operator and the real-time density matrix. 
\end{remark}

\subsection{Vectorization of the density matrix and representation of non-Markovian dynamics}
To analyze the above stochastic dynamics, we rewrite $\rho_c$ as a vector $\vec{\rho} = {\rm vec}\left[\rho\right]$ with ${\rm vec}[\bullet]$ representing the vectorization of a matrix. We assume that the dimension of the density matrix $\rho$ is $n \times n$, and that of $\vec{\rho}$ is $n^2\times 1$.  According to the mathematical relationship that ${\rm vec}[ABC] = \left(C^{\rm T}\otimes A\right){\rm vec}[B]$ with ${\rm T}$ for transpose, Eq.~(\ref{con:SME}) can be rewritten as the Liouville equation in vector form, namely
\begin{equation} \label{con:SMEVec}
\begin{aligned}
d\vec{\rho}   =&\mathbb{A}(t) \vec{\rho} dt  +\mathbb{B}\left(t, \rho\right) \vec{\rho} d W_t \triangleq \mathcal{L}\left(t,dW_t\right) \vec{\rho},
\end{aligned}
\end{equation}
with
\begin{subequations}   \label{con:matrixAB}
\begin{numcases}{}
\mathbb{A}(t) =  i \left( H^{\rm T} \otimes \mathbb{I}_n   - \mathbb{I}_n\otimes H\right) + O^{*} \otimes L + L^{*} \otimes O \notag\\
~~~~~~~~~-L^{\rm T} O^{*} \otimes \mathbb{I}_n -\mathbb{I}_n  \otimes  L^{\dag}O  ,\\
\mathbb{B}\left(t, \rho\right) = \mathbb{I}_n \otimes M + M^{*} \otimes \mathbb{I}_n  - {\rm Tr}\left[\left(M+M^{\dag}\right)\rho \right] \mathbb{I}_{n^2} ,
\end{numcases}
\end{subequations}
$\mathbb{I}_n$ is the identity matrix with the same dimension as $\rho$, and $\mathbb{I}_{n^2}$ represents that with the dimension $n^2\times n^2$~\cite{ture2024application}.

According to \cite{scopa2019exact}, for the $n \times n$ density matrix, the Hamiltonian $H$ and the operator $L$ in Eq.~(\ref{con:matrixAB}) can be represented with a complete set of unitary operators $\left\{F_j\right\}_{j=0}^{n^2-1} $ that satisfy $F_0 = \mathbb{I}_n/n$, $F_j = F_j^{\dag}$, ${\rm Tr} \left( F_j \right) = \delta_{j0}$,  ${\rm Tr} \left( F_j^{\dag} F_l \right)= \delta_{jl}$, and
\[ 
\delta_{jl} = \begin{cases} 
1, & j = l, \\ 
0, & j\neq l. 
\end{cases}
\]
Then the Lie algebra basis in the Liouville space can be represented based on $\left\{F_j\right\}_{j=0}^{n^2-1} $ as

\begin{subequations}   \label{con:LiouvilleLieBasis}
\begin{numcases}{}
\mathcal{H}_j = -i \left(\mathbb{I}_n \otimes F_j -F_j^* \otimes \mathbb{I}_n  \right),\\
\mathcal{D}_{jl} = F_l^* \otimes F_j - \frac{1}{2} \mathbb{I}_n \otimes F_lF_j - \frac{1}{2}  F_j^*F_l^* \otimes\mathbb{I}_n,
\end{numcases}
\end{subequations}
where $j,l = 1,2,\cdots, n^2-1$. It should be noted that in Eq.~(\ref{con:LiouvilleLieBasis}), the set $\left\{F_j\right\}_{j=0}^{n^2-1} $ generates a closed $\mathfrak{su}(n)$ Lie algebra, and $\mathcal{D}_{jl}$ are non-unitary operators.

The dynamics in Eq.~(\ref{con:SMEVec}) is on a complex finite-dimensional manifold $\mathcal{M}$, and $\vec{\rho} \in \mathcal{M}$. $\mathcal{G}$ denotes a finite-dimensional Lie group; then the Lie group action is regarded as $\Phi : \mathcal{G} \times \mathcal{M} \rightarrow  \mathcal{M}$, with the vector field $V = \mathcal{L}(t,dW) \vec{\rho}$ as in Eq.~(\ref{con:SMEVec})\cite{StochasticIntegrator}. 
For the evolution of the vector $\vec{\rho}$ in the time domain, we define the flow map $\varphi_t \left[\bullet \right] : \mathcal{M} \rightarrow  \mathcal{M}$, and 
\begin{equation} \label{con:flowmap}
\begin{aligned}
\vec{\rho}(t) = \varphi_t \left[\vec{\rho}(0)\right],
\end{aligned}
\end{equation}
by integrating Eq.~(\ref{con:SMEVec}).

We denote the exponential map between the Lie group $\mathcal{G}$ and Lie algebra $\mathfrak{g}$ as
\[
\mathrm{exp} : \mathfrak{g} \rightarrow \mathcal{G}.
\]
Then the quantum stochastic dynamics in Eq.~(\ref{con:SMEVec}) can be regarded as stochastic dynamics on Lie groups with the closed Lie algebra of superoperators in Eq.~(\ref{con:LiouvilleLieBasis}).  $\mathbb{A}(t)$ and $\mathbb{B}(t)$ in Eq.~(\ref{con:matrixAB}) can be represented as time-varying combinations of the basis $\mathcal{H}_j$ and $\mathcal{D}_{jl}$ in Eq.~(\ref{con:LiouvilleLieBasis}). That is,

\begin{subequations}   \label{con:ABrepresentation}
\begin{numcases}{}
\mathbb{A}(t) = \sum_j \hat{a}_j(t) \mathcal{H}_j + \sum_{j,l} \check{a}_{jl}(t)\mathcal{D}_{jl} ,\\
\mathbb{B}\left(t, \rho\right) =  \sum_j \hat{b}_j\left(t, \rho\right) \mathcal{H}_j + \sum_{j,l} \check{b}_{jl}\left(t, \rho\right)\mathcal{D}_{jl},
\end{numcases}
\end{subequations}
where the time-invariant components in $\hat{a}_j(t)$ and $\check{a}_{jl}(t)$ are due to the representation of the quantum system's free Hamiltonian in the Liouville space, the time-varying components in $\hat{a}_j(t)$ and $\check{a}_{jl}(t)$ are determined by the non-Markovian interactions between the quantum system and environment, the parameters $\hat{b}_j\left(t, \rho\right)$ and $\check{b}_{jl}\left(t, \rho\right)$ are determined by the filtering process of quantum systems, which is related to the real-time volution of quantum states and the measurement operators.
\subsection{Simplified linear dynamics}
For the simplified case without considering the measurement process by taking $\mathbb{B}\left(t, \rho\right) \equiv 0$, Eq.~(\ref{con:SMEVec}) reduces to a linear time-varying system
\begin{equation} \label{con:SMEVecNostochastic}
\begin{aligned}
\dot{\vec{\rho}} (t)  =&\mathbb{A}(t) \vec{\rho}(t) ,
\end{aligned}
\end{equation}
where the time-varying parameters are due to the quantum system's non-Markovian interactions with the environment.  

To derive the Magnus expansion of $\vec{\rho}(t)$, we assume that the initial state is $\vec{\rho}(0)$, and the quantum state after evolution can be represented as
\begin{equation} \label{con:rhoct}
\begin{aligned}
\vec{\rho}(t) = e^{\Lambda(t)} \vec{\rho}(0),
\end{aligned}
\end{equation}
where $\Lambda(t)$ is time-varying complex-value matrix with $\Lambda(0) = \mathbf{0}$. Then
\begin{equation} \label{con:drhoct}
\begin{aligned}
\dot{\vec{\rho}}(t) &= \frac{d e^{\Lambda(t)}}{dt} \vec{\rho}(0) = {\rm dexp}_{\Lambda(t)} \left[ \dot{\Lambda}(t) \right]\vec{\rho}(t),
\end{aligned}
\end{equation}
where 
\[
e^{\Lambda(t)} = \sum_{j = 0}^{\infty} \frac{\Lambda^j(t)}{j!} ,
\]
and
\begin{equation} 
\begin{aligned}
 \frac{d e^{\Lambda(t)}}{dt}&= \sum_{j=1}^{\infty} \frac{1}{j!}\frac{d\left[\Lambda(t) \right]^j}{dt}.
\end{aligned}
\end{equation}

As solved in \cite{SIAM2003magnus}, $\Lambda(t)$ is determined by the integral $\int_0^t \mathbb{A}(\tau) d\tau$ and other higher-order commutators, namely
\begin{equation} \label{con:LambdaDeffi}
\begin{aligned}
\frac{d \Lambda}{d t} &=\sum_{j=0}^{\infty} \frac{B_j}{j!} {\rm ad}_{\Lambda}^j(\mathbb{A})  = \mathbb{A} -\frac{1}{2} \left[ \Lambda,\mathbb{A}\right] + \frac{1}{12}\left[\Lambda,\left[ \Lambda,\mathbb{A}\right] \right]+\cdots,
\end{aligned}
\end{equation}
where $ {\rm ad}_{\Lambda}^j(\bullet)$ represents the $j$th order adjoint operator and the parameters $B_j$ are Bernoulli numbers. Then $\Lambda(t)$ can be represented as~\cite{blanes2009magnus}
\begin{equation} \label{con:LambdaDynamics}
\begin{aligned}
\Lambda(t)  =& \sum_{j=1}^{\infty}\Lambda_j(t)\\
=& \int_0^{t} \mathbb{A}(\tau)d\tau - \frac{1}{2} \int_0^{t}  \left[\int_0^{\tau} \mathbb{A}\left(\tau_1\right) d\tau_1, \mathbb{A}(\tau) \right]d\tau + \frac{1}{6}  \int_0^{t}\int_0^{\tau} \int_0^{\tau_1} \\
&\left(\left[\mathbb{A}\left(\tau\right) ,\left[  \mathbb{A}\left(\tau_1\right), \mathbb{A}\left(\tau_2\right)\right] \right] +\left[\mathbb{A}\left(\tau_2\right) ,\left[  \mathbb{A}\left(\tau_1\right), \mathbb{A}(\tau)\right] \right]\right)d\tau_2 d\tau_1d\tau +\cdots,
\end{aligned}
\end{equation}
where the mathematical formats of higher-order terms are omitted. Based on this, the truncation error of $p$th order Magnus expansion can be evaluated as
\begin{equation} \label{con:errordef}
\begin{aligned}
e_p =  \Lambda(t) -\sum_{j=1}^{p} \Lambda_j(t) \approx  \left| \Lambda_{p+1}(t) \right|. 
\end{aligned}
\end{equation}

For the simplified case with a time-invariant Hamiltonian and  Markovian integral kernel as in Eq.~(\ref{con:MarkovAlpha}),  $\left[ \mathbb{A}\left(\tau_1\right), \mathbb{A}(\tau) \right] \equiv 0$ and the solution of Eq.~(\ref{con:LambdaDynamics}) is $\Lambda(t)  = \mathbb{A}t$.

However, for the non-Markovian case, the Magnus integrator is determined by algebraic commutators such as $\left[\mathbb{A}\left(\tau_1\right) , \mathbb{A}(\tau) \right]$ and other higher-order terms in Eq.~(\ref{con:LambdaDynamics}). To further analyze this, we rewrite $\mathbb{A}(t)$ by separating it into time-invariant and time-varying parts as
\begin{equation} \label{con:Atcombine}
\begin{aligned}
\mathbb{A}(t) = \mathbb{H}+ V(t),
\end{aligned}
\end{equation}
where $\mathbb{H} =  i \left( H^{\rm T} \otimes \mathbb{I}_n   - \mathbb{I}_n\otimes H\right)$ is for the coherent evolution of the quantum system, and $V(t) = O^{*} \otimes L + L^{*} \otimes O -L^{\rm T} O^{*} \otimes \mathbb{I}_n -\mathbb{I}_n  \otimes  L^{\dag}O$ represents the dissipation to the environment. In the following, we first study the first-order Magnus expansion, then generalize to higher-order circumstances.

\subsubsection{First-order Magnus expansion}
Consider the integral in $\left[ t_0,t_0+h\right]$, according to Eq.~(\ref{con:LambdaDynamics}), 

\begin{subequations}   \label{con:Omega12}
\begin{numcases}{}
\Lambda_1 = \int_{t_0}^{t_0+h} \mathbb{A}(\tau)d\tau ,\label{Omega1}\\
\Lambda_2= - \frac{1}{2} \int_{t_0}^{t_0+h} \int_{t_0}^{\tau} \left[ \mathbb{A}\left(\tau_1\right) , \mathbb{A}(\tau) \right]d\tau_1 d\tau ,\label{Omega2}
\end{numcases}
\end{subequations}
then the truncation error for the first-order Magnus expansion can be evaluated by the Frobenius norm as
\begin{equation} \label{con:UpperboundO2}
\begin{aligned}
\left\|\Lambda_2 \right\| & \leq \frac{1}{2}\int_{t_0}^{t_0+h} \int_{t_0}^{\tau} \left\|\left[ \mathbb{A}\left(\tau_1\right) , \mathbb{A}(\tau) \right] \right\| d\tau_1 d\tau \\
& \leq \frac{h^2}{4} \max_{\tau,\tau_1\in\left[t_0,t_0+h \right]}  \left\|\left[ \mathbb{A}\left(\tau\right) , \mathbb{A}\left(\tau_1\right) \right] \right\|.
\end{aligned}
\end{equation}

For the upper bound in Eq.~(\ref{con:UpperboundO2}), the Frobenius norm of $\mathbb{A}(t)$ satisfies that 
\begin{equation} \label{con:Atnorm}
\begin{aligned}
\left\|\mathbb{A}(t)\right\| &\leq \left\|\mathbb{H}\right\| + \left\|V(t)\right\| \\
&\leq 2\sqrt{n} \left\|H\right\| + 2 \left\|L\right\| \left\|O\right\| + 2\sqrt{n}\left\|L\right\| \left\|O\right\|\\
& = 2\sqrt{n} \left\|H\right\| + 2 \left(1+\sqrt{n} \right)\left\|L\right\| \left\|O\right\|,
\end{aligned}
\end{equation}
where due to the definition of operator $O$ in Eq.~(\ref{con:MasterEq}), 
\begin{equation} 
\begin{aligned}
\left\|O\right\| &\leq \sum_{j=1}^N \left| f_j(t) \right| \left\|L_j\right\| .
\end{aligned}
\end{equation}

Obviously, for reduced Markovian dynamics, the truncation error is equal to zero because $\mathbb{A}$ is time-invariant and $\left[ \mathbb{A}\left(\tau\right) , \mathbb{A}\left(\tau_1\right) \right]\equiv 0$. However, for general non-Markovian dynamics, according to the combination of time-invariant and time-varying components in Eq.~(\ref{con:Atcombine}), 
\begin{equation}  \label{con:ErrorCommutator}
\begin{aligned}
 \left\|\left[ \mathbb{A}\left(\tau\right) , \mathbb{A}\left(\tau_1\right) \right] \right\| &=  \left\|\left[  \mathbb{H}+ V\left(\tau\right) ,  \mathbb{H}+ V\left(\tau_1\right) \right] \right\|\\
 & =  \left\|\left[  \mathbb{H},  V\left(\tau_1\right) -V\left(\tau\right)   \right] + \left[V\left(\tau\right),V\left(\tau_1\right) \right]\right\|\\
 & \leq  \left\|\left[  \mathbb{H} ,  V\left(\tau_1\right) -V\left(\tau\right)   \right] \right\| + \left\| \left[V\left(\tau\right),V\left(\tau_1\right) \right]\right\|\\
 & \leq \sqrt{2} \left\|\mathbb{H} \right\| \left\| V\left(\tau_1\right) -V\left(\tau\right)\right\| + \sqrt{2} \left\|V\left(\tau\right) \right\| \left\| V\left(\tau_1\right)\right\|.
\end{aligned}
\end{equation}

Additionally, by further analyzing the mathematical formats of $V(t)$, the upper bound above can be evaluated more precisely. To clarify this, we  consider the evolution around the middle time point $t_m = t_0+h/2$, then
\begin{equation} 
\begin{aligned}
 \mathbb{A}\left(t\right) =  \mathbb{A}\left(t_m\right) + \dot{\mathbb{A}}\left(t_m\right)\left(t-t_m\right) + \mathcal{O}\left(h^2\right),
\end{aligned}
\end{equation}
where the last term represents the infinitesimal of higher orders. Then the commutator in Eq.~(\ref{con:ErrorCommutator}) can be further simplified as
\begin{equation}  \label{con:ErrorCommutatorMiddle}
\begin{aligned}
&\left[ \mathbb{A}\left(\tau\right) , \mathbb{A}\left(\tau_1\right) \right] \\
=& \left[  \mathbb{A}\left(t_m\right) + \dot{\mathbb{A}}\left(t_m\right)\left(\tau-t_m\right) + \mathcal{O}\left(h^2\right) ,   \mathbb{A}\left(t_m\right) + \dot{\mathbb{A}}\left(t_m\right)\left(\tau_1-t_m\right) + \mathcal{O}\left(h^2\right)  \right] \\
=& \left[  \mathbb{A}\left(t_m\right) + \dot{V}\left(t_m\right)\left(\tau-t_m\right) + \mathcal{O}\left(h^2\right) ,   \mathbb{A}\left(t_m\right) + \dot{V}\left(t_m\right)\left(\tau_1-t_m\right) + \mathcal{O}\left(h^2\right)  \right]\\
=& \left[ \mathbb{A}\left(t_m\right),\dot{V}\left(t_m\right)\right]\left(\tau_1-\tau\right) + \mathcal{O}\left(h^2\right) .
\end{aligned}
\end{equation}
Based on this, the integrator in Eq.~(\ref{Omega2}) can be rewritten as
\begin{equation}  \label{con:Omega2Middle}
\begin{aligned}
\Lambda_2&= \frac{1}{2} \int_{t_0}^{t_0+h} \int_{t_0}^{\tau} \left[\mathbb{A}(\tau)  ,\mathbb{A}\left(\tau_1\right)\right]d\tau_1 d\tau\\
& \approx \frac{1}{2}\left[ \mathbb{A}\left(t_m\right),\dot{V}\left(t_m\right)\right] \int_{t_0}^{t_0+h} \int_{t_0}^{\tau} \left(\tau_1-\tau\right)  d\tau_1 d\tau\\
& = -\frac{1}{12}\left[ \mathbb{A}\left(t_m\right),\dot{V}\left(t_m\right)\right] h^3.
\end{aligned}
\end{equation}

For the truncation error in the last line of Eq.~(\ref{con:ErrorCommutator}), we rewrite $V(t)$ as
\begin{equation} \label{con:VtRewrite}
\begin{aligned}
V(t)=&\sum_{j =1}^N \left[ f_j(t) \left(   L^{*} \otimes L_j - \mathbb{I}_n  \otimes  L^{\dag}L_j \right) + f_j^*(t) \left( L_j \otimes L-L^{\rm T} L_j \otimes \mathbb{I}_n\right)\right]\\
\triangleq & \sum_{j=1}^{2N} \Gamma_j(t) V_j,
\end{aligned}
\end{equation}
where $\Gamma_j(t) = f_j(t)$, $\Gamma_{j+N}(t) = f_j^*(t)$, and the operators

\begin{subequations}   \label{con:V1to4}
\begin{numcases}{}
V_j =  L^{*} \otimes L_j - \mathbb{I}_n  \otimes  L^{\dag}L_j,\\
V_{j+N} =   L_j \otimes L-L^{\rm T} L_j \otimes \mathbb{I}_n.
\end{numcases}
\end{subequations}
Then the truncation error can be evaluated according to the following theorem.
\begin{theorem} \label{FirstOrderError}
When solving  the Liouville equation (\ref{con:SMEVecNostochastic}) with Magnus expansion, given $V(t)$ in Eq.~(\ref{con:VtRewrite}) for the non-Markovian interaction between the quantum system and environment, the truncation error of the first-order expansion can be evaluated as
\begin{equation} \label{con:errorlemme}
\begin{aligned}
e_1\leq \frac{h^3}{12} \left( \sum_{j=1}^{2N} \left|\dot{\Gamma}_j\left(t_m\right) \right|\left\| \left[\mathbb{H}, V_j \right]\right\| + \sum_{j,l=1,j\neq l}^{2N} \left| \Gamma_l\left(t_m\right) \dot{\Gamma}_j\left(t_m\right) \right| \left\| \left[V_l, V_j \right]\right\| \right) .
\end{aligned}
\end{equation}
\end{theorem}
\begin{proof}
The truncation error for the first-order Magnus expansion can be determined by Eq.~(\ref{con:Omega2Middle}). Because $\left[\Gamma_j(t)V_j,\dot{\Gamma}_j(t)V_j \right] =0$, the truncation error can be simplified as Eq.~(\ref{con:errorlemme}).
\end{proof}

\subsubsection{Second-order Magnus expansion}
For second-order Magnus expansions based on Eq.~(\ref{con:LambdaDynamics}), the truncation error can be evaluated by third-order commutators. Generalized from Eq.~(\ref{con:ErrorCommutatorMiddle}), the following commutator can be represented around the time point $t_m$ as
\begin{equation}  \label{con:TwoOrderPart1}
\begin{aligned}
&\left[\mathbb{A}\left(\tau\right) ,\left[  \mathbb{A}\left(\tau_1\right), \mathbb{A}\left(\tau_2\right)\right] \right]\\
= &\left[\mathbb{A}\left(\tau\right) ,\left[\mathbb{A}\left(t_m\right), \dot{V}\left(t_m\right) \right]\left(\tau_2-\tau_1 \right) +  \mathcal{O}\left(h^2\right) \right]\\
=&\left[\mathbb{A}\left(t_m\right) + \dot{V}\left(t_m\right)\left(\tau-t_m\right) ,\left[\mathbb{A}\left(t_m\right), \dot{V}\left(t_m\right) \right]\left(\tau_2-\tau_1 \right) +  \mathcal{O}\left(h^2\right) \right]\\
=&\left(\tau_2-\tau_1 \right) \left[\mathbb{A}\left(t_m\right) ,\left[\mathbb{A}\left(t_m\right), \dot{V}\left(t_m\right) \right]  \right] \\
&+\left(\tau-t_m\right) \left(\tau_2-\tau_1 \right) \left[\dot{V}\left(t_m\right) ,\left[\mathbb{A}\left(t_m\right), \dot{V}\left(t_m\right) \right] \right] +  \mathcal{O}\left(h^2\right),
\end{aligned}
\end{equation}
and similarly, 
\begin{equation}  \label{con:TwoOrderPart2}
\begin{aligned}
&\left[\mathbb{A}\left(\tau_2\right) ,\left[  \mathbb{A}\left(\tau_1\right), \mathbb{A}\left(\tau\right)\right] \right]\\
=&\left(\tau-\tau_1 \right) \left[\mathbb{A}\left(t_m\right) ,\left[\mathbb{A}\left(t_m\right), \dot{V}\left(t_m\right) \right]  \right] \\
&+\left(\tau_2-t_m\right) \left(\tau-\tau_1 \right) \left[\dot{V}\left(t_m\right) ,\left[\mathbb{A}\left(t_m\right), \dot{V}\left(t_m\right) \right] \right] +  \mathcal{O}\left(h^2\right).
\end{aligned}
\end{equation}
By summarizing Eq.~(\ref{con:TwoOrderPart1})  and Eq.~(\ref{con:TwoOrderPart2}), we can derive
\begin{equation}  \label{con:TwoOrderSum}
\begin{aligned}
&\left[\mathbb{A}\left(\tau\right) ,\left[  \mathbb{A}\left(\tau_1\right), \mathbb{A}\left(\tau_2\right)\right] \right] + \left[\mathbb{A}\left(\tau_2\right) ,\left[  \mathbb{A}\left(\tau_1\right), \mathbb{A}\left(\tau\right)\right] \right]\\
=&\left(\tau + \tau_2-2\tau_1 \right) \left[\mathbb{A}\left(t_m\right) ,\left[\mathbb{A}\left(t_m\right), \dot{V}\left(t_m\right) \right]  \right] \\
&+t_m\left(2\tau_1 - \tau-\tau_2 \right) \left[\dot{V}\left(t_m\right) ,\left[\mathbb{A}\left(t_m\right), \dot{V}\left(t_m\right) \right] \right] \\
&+\left[ \tau \left(\tau_2-\tau_1 \right)+ \tau_2 \left(\tau-\tau_1 \right)\right]  \left[\dot{V}\left(t_m\right) ,\left[\mathbb{A}\left(t_m\right), \dot{V}\left(t_m\right) \right] \right] +  \mathcal{O}\left(h^2\right).
\end{aligned}
\end{equation}

Then we have the following theorem on the truncation error of the second-order Magnus expansion.
\begin{theorem} \label{SecondOrderError}
When solving  Eq.~(\ref{con:SMEVecNostochastic}) with Magnus expansion, the truncation error of the second-order expansion can be evaluated as
\begin{equation} \label{con:errorlemmeTwoOrder}
\begin{aligned}
e_2 = \frac{h^5}{240} \left\|\left[\sum_{k=1}^{2N} \dot{\Gamma}_k\left(t_m\right) V_k ,\left[\mathbb{H}, \sum_{j=1}^{2N} \dot{\Gamma_j}\left(t_m\right) V_j \right] + \sum_{j,l=1,j\neq l}^{2N}  \Gamma_l\left(t_m\right) \dot{\Gamma}_j\left(t_m\right)  \left[V_l, V_j \right] \right] \right\| .
\end{aligned}
\end{equation}
\end{theorem}
\begin{proof}
For the third-order term in Eq.~(\ref{con:LambdaDynamics}), we consider the following integral around the middle time point $t_m$, 
\begin{equation}  \label{con:Omega3Middle}
\begin{aligned}
&\Lambda_3= \frac{1}{6} \int_{t_0}^{t_0+h} \int_{t_0}^{\tau}\int_{t_0}^{\tau_1} \left(\left[\mathbb{A}\left(\tau\right) ,\left[  \mathbb{A}\left(\tau_1\right), \mathbb{A}\left(\tau_2\right)\right] \right] +\left[\mathbb{A}\left(\tau_2\right) ,\left[  \mathbb{A}\left(\tau_1\right), \mathbb{A}(\tau)\right] \right]\right)d\tau_2 d\tau_1d\tau\\
& \approx \frac{1}{6}\left[\dot{V}\left(t_m\right) ,\left[\mathbb{A}\left(t_m\right), \dot{V}\left(t_m\right) \right] \right] \int_{t_0}^{t_0+h} \int_{t_0}^{\tau}\int_{t_0}^{\tau_1} \left[ \tau \left(\tau_2-\tau_1 \right)+ \tau_2 \left(\tau-\tau_1 \right)\right]d\tau_2 d\tau_1d\tau  \\
& = -\frac{h^5}{240}\left[\dot{V}\left(t_m\right) ,\left[\mathbb{A}\left(t_m\right), \dot{V}\left(t_m\right) \right] \right] ,
\end{aligned}
\end{equation}
where the other terms are equal to zero because $\tau + \tau_2-2\tau_1$ is an odd function of $\tau,\tau_1,\tau_2$ in Eq.~(\ref{con:TwoOrderSum}). Further combined with Eq.~(\ref{con:VtRewrite}), the truncation error can be derived as Eq.~(\ref{con:errorlemmeTwoOrder}).
\end{proof}

\begin{remark}
Theorem~\ref{FirstOrderError} and Theorem~\ref{SecondOrderError} illustrate that the first- and second-order truncation errors are determined by the Lie algebras generated by $\mathbb{H}$ for the quantum system's free Hamiltonian and $\Gamma_j(t)V_j$ for the quantum system's time-varying non-Markovian dissipations to the environment. 
\end{remark}

\subsection{Convergence analysis on the simplified linear dynamics} Now we analyze the convergence of the Magnus expansion based on the non-Markovian dynamics in Eq.~(\ref{con:LambdaDynamics}). Due to the flow map in Eq.~(\ref{con:flowmap}), $\vec{\rho}(t)$ can be represented as the product of its initial value and a matrix $e^{\Lambda(t)}$ with dimension $n^2\times n^2$.
According to \cite{casas2007sufficient}, the convergence of the Magnus expansion means that the infinite series in Eq.~(\ref{con:LambdaDynamics}) converges and the real-time quantum state can be solved by Eq.~(\ref{con:rhoct}).

Based on Eq.~(\ref{con:LambdaDynamics}), when the following condition is satisfied, 
\begin{equation} \label{con:ConvergenceCondition}
\begin{aligned}
\int_0^t \left\|\mathbb{A}(\tau) \right\|_2 d\tau < \pi,
\end{aligned}
\end{equation}
where $\left\| \bullet \right\|_2$ represents the 2-norm or spectral norm, the Magnus expansion converges~\cite{casas2007sufficient,moan2008convergence}. This has been adopted to analyze the Markovian dynamics of a two-level system in \cite{begzjav2020magnus}.

Generalized from the model in Eq.~(\ref{con:MasterEq}), we rewrite the Hamiltonian as

\begin{equation} 
\begin{aligned}
H = H_0 + H_{\rm C}(t),
\end{aligned}
\end{equation}
where $H_0$ is the time-invariant free Hamiltonian, and $H_{\rm C}(t)$ is the time-varying control Hamiltonian. We denote the rotational operator $U_0(t) = e^{-iH_0t}$, the density matrix for the quantum state represented in the interaction picture reads $\rho_{\rm I}(t) = U_0^{\dag}(t)\rho(t) U_0(t)$, and the operators representing the interaction between quantum system and environment are 

\begin{subequations}   \label{con:OperatorInteraction}
\begin{numcases}{}
L_{\rm I}(t) = U_0^{\dag}(t)L U_0(t) ,\\
O_{\rm I}(t) = U_0^{\dag}(t)O U_0(t). \label{OI}
\end{numcases}
\end{subequations}

According to the derivations in Appendix ~\ref{Sec:IntractionPicture}, the non-Markovian equation in the interaction picture can be represented as
\begin{equation} \label{con:dRhoInt}
\begin{aligned}
\frac{d\rho_{\rm I}}{dt} =-i \left[ H_{\rm CI}(t) ,\rho_{\rm I} \right]+ \left[ L_{\rm I}(t), \rho_{\rm I}  O_{\rm I}^{\dag}(t)  \right] + \left[ O_{\rm I}(t) \rho_{\rm I}, L_{\rm I}^{\dag}(t)\right],
\end{aligned}
\end{equation}
where $H_{\rm CI} (t)= U_0^{\dag}(t)H_{\rm C} U_0(t)$. We denote $\vec{\rho}_{\rm I} = {\rm vec}\left[\rho_{\rm I}\right]$, and similar to Eq.~(\ref{con:SMEVecNostochastic}),
\begin{equation} \label{con:rhoIvec}
\begin{aligned}
\frac{d\vec{\rho}_{\rm I}}{dt} = V_{\rm I}(t) \vec{\rho}_{\rm I} ,
\end{aligned}
\end{equation}
where $V_{\rm I}(t) =i \left( H_{\rm CI}^{\rm T} \otimes \mathbb{I}_n   - \mathbb{I}_n\otimes H_{\rm CI}\right)  +  O_{\rm I}^{*} \otimes L_{\rm I} + L_{\rm I}^{*} \otimes O_{\rm I} -L_{\rm I}^{\rm T} O_{\rm I}^{*} \otimes \mathbb{I}_n -\mathbb{I}_n  \otimes  L_{\rm I}^{\dag} O_{\rm I}$, and $O_{\rm I}(t)  =\sum_{j=1}^Nf_j(t) U_0^{\dag}(t)L_jU_0(t)$.

\begin{theorem} \label{convergence}
In the interaction picture, when 
\begin{equation}  \label{con:InterPcondtion}
\begin{aligned}
\int_0^t \left\|V_{\rm I}(\tau) \right\|_2 d\tau < \pi,
\end{aligned}
\end{equation}
the Magnus expansion for solving Eq.~(\ref{con:rhoIvec}) converges. If the control Hamiltonian satisfies $\left\| H_{\rm C}\right\|_2 \ll\left\| H_0 \right\|_2 $, then
\begin{equation} \label{con:Inequality}
\begin{aligned}
{\rm{sup}}\left(\int_0^t \left\|V_{\rm I}(\tau) \right\|_2 d\tau \right) \leq {\rm{sup}}\left(\int_0^t \left\|\mathbb{A}(\tau) \right\|_2 d\tau \right),
\end{aligned}
\end{equation}
with $\mathbb{A}(t)$ given by Eq.~(\ref{con:SMEVecNostochastic}). 
\end{theorem}
\begin{proof}
The condition in Eq.~(\ref{con:InterPcondtion}) can be directly generalized and proved by Eq.~(\ref{con:ConvergenceCondition}) and \cite{casas2007sufficient,moan2008convergence}. Considering that in Eq.~(\ref{con:Atcombine})
\begin{equation} \label{con:AWV}
\begin{aligned}
\int_0^t \left\|\mathbb{A}(\tau) \right\|_2 d\tau  &\leq \int_0^t \left\|\mathbb{H} \right\|_2 d\tau + \int_0^t \left\|V(\tau)\right\|_2 d\tau\\
& \leq \int_0^t \left\|\mathbb{H} \right\|_2 d\tau + \sum_{j=1}^{2N} \int_0^t  \left|\Gamma_j(\tau)\right|  \left\| V_j \right\|_2 d\tau,
\end{aligned}
\end{equation}
according to the representation in Eq.~(\ref{con:VtRewrite}). 

For the circumstance in the interaction picture, consider the first component of $V_{\rm I}(t)$,
\begin{equation} \label{con:Provepart1}
\begin{aligned}
\left\| H_{\rm CI}^{\rm T} \otimes \mathbb{I}_n   - \mathbb{I}_n\otimes H_{\rm CI}\right\|_2 \leq \left\| H_{\rm C}^{\rm T} \otimes \mathbb{I}_n \right\|_2  + \left\| \mathbb{I}_n\otimes H_{\rm C}\right\|_2.
\end{aligned}
\end{equation}
For the following components of $V_{\rm I}(t)$,
\begin{equation} \label{con:VItnorm}
\begin{aligned}
& \left\|O_{\rm I}^{*}(t) \otimes L_{\rm I}(t) + L_{\rm I}^{*}(t) \otimes O_{\rm I}(t) -L_{\rm I}^{\rm T}(t) O_{\rm I}^{*} (t)\otimes \mathbb{I}_n -\mathbb{I}_n  \otimes  L_{\rm I}^{\dag} (t)O_{\rm I}(t) \right\|_2\\
\leq& \left\|O_{\rm I}^{*}(t) \otimes L_{\rm I}(t) \right\|_2  + \left\| L_{\rm I}^{*}(t) \otimes O_{\rm I}(t)\right\|_2\\
& + \left\|L_{\rm I}^{\rm T}(t) O_{\rm I}^{*} (t)\otimes \mathbb{I}_n \right\|_2  + \left\| \mathbb{I}_n  \otimes  L_{\rm I}^{\dag} (t)O_{\rm I}(t) \right\|_2.
\end{aligned}
\end{equation}
We take one representative example in Eq.~(\ref{con:VItnorm}) for clarification,  
\begin{equation} 
\begin{aligned}
\left\|O_{\rm I}^{*}(t) \otimes L_{\rm I}(t) \right\|_2  &= \left\|O_{\rm I}^{*}(t) \right\|_2 \left\| L_{\rm I}(t) \right\|_2 \\
& = \left\|O_{\rm I}(t) \right\|_2 \left\| L_{\rm I}(t) \right\|_2\\
& = \left\|U_0^{\dag}(t) \right\|_2 \left\|O \right\|_2 \left\|U_0(t) \right\|_2  \left\|U_0^{\dag}(t) \right\|_2 \left\|L \right\|_2 \left\|U_0(t) \right\|_2 \\
& =  \left\|O \right\|_2   \left\|L \right\|_2.
\end{aligned}
\end{equation}
Combined with other components in Eq.~(\ref{con:VItnorm}) with similar formats and the definition of $V_j $ in Eq.~(\ref{con:V1to4}), we have
\begin{equation} \label{con:Provepart2}
\begin{aligned}
\left\| V_{\rm I}(t) \right\|_2 &\leq \left\| H_{\rm C}^{\rm T} \otimes \mathbb{I}_n \right\|_2  + \left\| \mathbb{I}_n\otimes H_{\rm C}\right\|_2+ \sum_{j=1}^{2N} \left|\Gamma_j(t)\right|  \left\| V_j \right\|_2\\
&\leq \left\| \left( H_0 + H_{\rm C}\right)^{\rm T} \otimes \mathbb{I}_n \right\|_2  + \left\| \mathbb{I}_n\otimes  \left( H_0 + H_{\rm C}\right)\right\|_2+ \sum_{j=1}^{2N} \left|\Gamma_j(t)\right|  \left\| V_j \right\|_2,
\end{aligned}
\end{equation}
due to the condition that $\left\| H_{\rm C}\right\|_2 \ll\left\| H_0 \right\|_2 $. Based on Eq.~(\ref{con:Provepart1}) and Eq.~(\ref{con:Provepart2}), Eq.~(\ref{con:Inequality}) can be proved.
\end{proof}

Theorem~\ref{convergence} clarifies that when solving the non-Markovian quantum dynamics with Magnus expansions in the interaction picture, the convergence of the numerical method can be improved because the condition in Eq.~(\ref{con:InterPcondtion}) can be more easily satisfied. The conclusion also holds for the integral in $\left[ t_0,t_0+h\right]$.

\section{Stochastic filtering dynamics} \label{Sec:Stochastic} In this section, we study the {\rm{It\^{o}}} stochastic dynamics with non-Markovianity based on Eqs.~(\ref{con:SME},\ref{con:SMEVec})  and the generalized Stratonovich stochastic dynamics.
\subsection{{\rm{It\^{o}}} stochastic dynamics}
We first derive the Magnus expansion for the stochastic format by taking the nonlinear coefficient $\mathbb{B}\left(t, \rho\right) \neq 0$ in Eq.~(\ref{con:SMEVec}). For quantum filtering in $\left[t_0,t_0+h \right]$ with a small time step $h\gg dt$, the evolution of quantum state can be represented as 
\begin{equation} \label{con:SMEVect0}
\begin{aligned}
\vec{\rho}\left( t_0 + dt\right)   =& \vec{\rho}\left(t_0\right) + \mathbb{A}\left(t_0\right) \vec{\rho}\left(t_0\right)dt  +\mathbb{B}\left(t_0, \rho\left(t_0 \right)\right) \vec{\rho}\left(t_0 \right) d W_{t_0},
\end{aligned}
\end{equation} 
where $\mathbb{A}\left(t_0\right)$ and $\mathbb{B}\left(t_0, \rho\left(t_0 \right)\right)$ can be acquired according to Eq.~(\ref{con:matrixAB}).

Given the measurement information at $t_0$, the original nonlinear stochastic dynamics in Eq.~(\ref{con:SMEVec}) is reduced to be linear. Additionally, generalized from Eq.~(\ref{con:rhoct}), we represent the evolution of a quantum system as
\begin{equation} \label{con:rhotPdt}
\begin{aligned}
\vec{\rho}\left(t_0+h\right) = e^{\tilde{\Lambda} (h)} \vec{\rho}\left(t_0\right),
\end{aligned}
\end{equation}
where $e^{\tilde{\Lambda} (h)}$ depends not only on the time step $h$, but also on the initial condition at $t_0$. For arbitrarily $s\in \left[t_0,t_0+h \right]$,
\begin{equation} \label{con:PhisDef}
\begin{aligned}
\Phi_s = e^{\tilde{\Lambda} (s)} = \mathbb{I} + \tilde{\Lambda} (s) + \frac{1}{2!}\left[\tilde{\Lambda} (s) \right]^2 + \frac{1}{3!}\left[\tilde{\Lambda} (s) \right]^3+ \cdots ,
\end{aligned}
\end{equation}
then Eq.~(\ref{con:rhotPdt}) can also be represented as $\vec{\rho}\left(t_0+h\right) = \Phi_h  \vec{\rho}\left(t_0\right)$.

Similar to Eq.~(\ref{con:LambdaDynamics}), the generalization to the stochastic circumstance reads
\begin{equation} \label{con:StochasticLambda}
\begin{aligned}
\tilde{\Lambda} \left(t_0+h\right)  =& \sum_{j=1}^{\infty}\tilde{\Lambda}_j\left(t_0+h\right) ,
\end{aligned}
\end{equation}
where
\begin{equation}\label{con:StochasticLambdaFirstOrder}
\begin{aligned}
\tilde{\Lambda}_1\left(t_0+h\right)
&= \int_{t_0}^{t_0+h} \mathbb{A}\left(\tau\right)d\tau + \mathbb{B}_{t_0} \int_{t_0}^{t_0+h} \xi\left(\tau \right) d\tau,
\end{aligned}
\end{equation}
$\mathbb{A}\left(t\right)$ is independent of the measurement of quantum states and has the same format as that in Eq.~(\ref{con:Omega12}), $\mathbb{B}_{t_0}$ is short for $\mathbb{B}\left(t_0, \rho\left(t_0 \right)\right) $ in Eq.~(\ref{con:SMEVect0}) as
\begin{equation} \label{con:Btilde}
\begin{aligned}
\mathbb{B}_{t_0}= \mathbb{I}_n \otimes M + M^{*} \otimes \mathbb{I}_n  - {\rm Tr}\left[\left(M+M^{\dag}\right)\rho\left(t_0 \right) \right] \mathbb{I}_{n^2} ,
\end{aligned}
\end{equation}
by normalizing the quantum system according to its initial state $\rho\left(t_0\right)$. Additionally, for the last component of Eq.~(\ref{con:StochasticLambdaFirstOrder}),
\begin{equation} \label{con:Wh}
\begin{aligned}
\int_{t_0}^{t_0+h} \xi\left(\tau \right) d\tau = W_{t_0+h} -W_{t_0} \triangleq \Delta W_{t_0+h},
\end{aligned}
\end{equation}
where for the Brownian motion, 
\begin{equation}
\begin{aligned}
&{\rm Var} \left[ W_{t_0+h} -W_{t_0}  \right] \\
= &{\rm Var} \left[ W_{t_0+h} \right] + {\rm Var} \left[W_{t_0}  \right] - 2{\rm Cov} \left[ W_{t_0+h}, W_{t_0}  \right]\\
= &t_0+h  +t_0- 2t_0\\
=&h.
\end{aligned}
\end{equation}
Thus $\Delta W_{t_0+h} \sim \mathcal{N}(0,h)$ follows a normal distribution. Based on this, Eq.~(\ref{con:StochasticLambdaFirstOrder}) can be simplified as
\begin{equation}\label{con:StochasticFirstOrderSimple}
\begin{aligned}
\tilde{\Lambda}_1\left(t_0+h\right) &= \int_{t_0}^{t_0+h} \mathbb{A}\left(\tau\right)d\tau + \mathbb{B}_{t_0}  \Delta W_{t_0+h}.
\end{aligned}
\end{equation}

According to Eq.~(\ref{con:StochasticFirstOrderSimple}), for $t_0 < s < t_0+h$, we denote 
\begin{equation}\label{con:Delta1s}
\begin{aligned}
\tilde{\Lambda}_1(s) &= \int_{t_0}^{s} \mathbb{A}\left(\tau\right)d\tau + \mathbb{B}_{t_0}  \Delta W_s,
\end{aligned}
\end{equation}
as a stochastic process indexed by $s$, and its increment
\begin{equation} \label{con:dLambda1}
\begin{aligned}
d\left[\tilde{\Lambda}_1(s)\right] &=  \mathbb{A}\left(s\right)ds +\mathbb{B}_{t_0} dW_s.
\end{aligned}
\end{equation}

For the above stochastic differential increment, consider the differential of $\Phi_s$ in Eq.~(\ref{con:PhisDef}) with stochasticity, 
\begin{equation}\label{con:dPhis}
\begin{aligned}
d\Phi_s&  =\left[ e^{\tilde{\Lambda} (s) + d\left[ \tilde{\Lambda} (s)\right]} -e^{\tilde{\Lambda} (s)} \right]\\ 
&= D\left[\Phi_s \right] d\left[\tilde{\Lambda} (s) \right]+ \frac{1}{2}D^2\left[\Phi_s \right] \left(\tilde{\Lambda} (s),\tilde{\Lambda} (s) \right) + \mathcal{O}\left[\tilde{\Lambda} (s)^2 \right],
\end{aligned}
\end{equation}
where the first and second components in the second line of Eq.~(\ref{con:dPhis}) briefly represent the first-order and second-order Fr\'{e}chet derivatives, respectively, and the last term represents the higher-order infinitesimal. We further clarify the Fr\'{e}chet derivatives as follows.

\subsection{Stochastic derivatives and Magnus expansions}
Similar to Eq.~(\ref{con:dLambda1}),
\begin{equation}
\begin{aligned}
 d\left[ \tilde{\Lambda} (s)\right] = \sum_{j=1}^{\infty}d\left[\tilde{\Lambda}_j(s)\right].
\end{aligned}
\end{equation}
Mathematically, 
\begin{equation}
\begin{aligned}
e^{\tilde{\Lambda} (s) + d\left[ \tilde{\Lambda} (s)\right]} = \mathbb{I}_{n^2} + \sum_{j =1}^{\infty} \left\{\tilde{\Lambda} (s) + d\left[ \tilde{\Lambda} (s)\right]\right\}^j,
\end{aligned}
\end{equation}
and
\begin{equation}
\begin{aligned}
e^{\tilde{\Lambda} (s) } = \mathbb{I}_{n^2} + \sum_{j =1}^{\infty} \left[\tilde{\Lambda} (s) \right]^j,
\end{aligned}
\end{equation}
then in Eq.~(\ref{con:dPhis}),
\begin{equation}\label{con:dPhisV2}
\begin{aligned}
&d\Phi_s =\sum_{j =1}^{\infty} \left\{\tilde{\Lambda} (s) + d\left[ \tilde{\Lambda} (s)\right]\right\}^j - \sum_{j =1}^{\infty} \left[\tilde{\Lambda} (s) \right]^j\\
&= d\left[ \tilde{\Lambda} (s)\right] +\frac{1}{2}\left\{ \tilde{\Lambda} (s) d\left[ \tilde{\Lambda} (s)\right] + d\left[ \tilde{\Lambda} (s)\right]\tilde{\Lambda} (s) + \left\{d\left[ \tilde{\Lambda} (s)\right]\right\}^2 \right\}+ \cdots\\
&= d\left[ \tilde{\Lambda} (s)\right] +\frac{1}{2}\left[ \tilde{\Lambda} (s) ,d\left[ \tilde{\Lambda} (s)\right]\right]+ \frac{1}{2}\left\{d\left[ \tilde{\Lambda} (s)\right]\right\}^2+ d\left[ \tilde{\Lambda} (s)\right] \tilde{\Lambda} (s)+ \cdots.
\end{aligned}
\end{equation}
Considering that
\begin{equation}
\begin{aligned}
&\Phi_s^{-1}  = e^{-\tilde{\Lambda} (s)} = \mathbb{I} - \tilde{\Lambda} (s)+ \cdots,
\end{aligned}
\end{equation}
then
\begin{equation}
\begin{aligned}
d\Phi_s \Phi_s^{-1} &= d\left[ \tilde{\Lambda} (s)\right] +\frac{1}{2}\left[ \tilde{\Lambda} (s) ,d\left[ \tilde{\Lambda} (s)\right]\right]+ d\left[ \tilde{\Lambda} (s)\right] \tilde{\Lambda} (s)\\
&+ \frac{1}{2}\left\{d\left[ \tilde{\Lambda} (s)\right]\right\}^2 -   d\left[ \tilde{\Lambda} (s)\right] \tilde{\Lambda} (s) - \frac{1}{2}\left[ \tilde{\Lambda} (s) ,d\left[ \tilde{\Lambda} (s)\right]\right]  \tilde{\Lambda} (s)  \\
&-   d\left[ \tilde{\Lambda} (s)\right] \tilde{\Lambda} (s)  \tilde{\Lambda} (s) - \frac{1}{2}\left\{d\left[ \tilde{\Lambda} (s)\right]\right\}^2  \tilde{\Lambda} (s) + \cdots\\
&\approx d\left[ \tilde{\Lambda} (s)\right] +\frac{1}{2}\left[ \tilde{\Lambda} (s) ,d\left[ \tilde{\Lambda} (s)\right]\right]+ \frac{1}{2}\left\{d\left[ \tilde{\Lambda} (s)\right]\right\}^2. 
\end{aligned}
\end{equation}
This agrees with the first- and second-order  Fr\'{e}chet derivatives in Eq.~(\ref{con:dPhis}), namely
\begin{equation}
\begin{aligned}
d\Phi_s  =&  \left\{ d\left[ \tilde{\Lambda} (s)\right] +\frac{1}{2}\left[ \tilde{\Lambda} (s) ,d\left[ \tilde{\Lambda} (s)\right]\right]+\frac{1}{6}\left[ \tilde{\Lambda} (s) , \left[ \tilde{\Lambda} (s) ,d\left[ \tilde{\Lambda} (s)\right]\right] \right]\right. \\
&\left. +\cdots+\frac{1}{2}\left\{d\left[ \tilde{\Lambda} (s)\right]\right\}^2  + \frac{1}{6} \left[  \tilde{\Lambda} (s) ,\left\{d\left[ \tilde{\Lambda} (s)\right]\right\}^2 \right]\right.\\
&\left. + \frac{1}{12}\left[d\left[ \tilde{\Lambda} (s)\right],\left[ \tilde{\Lambda} (s) ,d\left[ \tilde{\Lambda} (s)\right] \right]\right] + \dots\right\}\Phi_s. 
\end{aligned}
\end{equation}
Based on Eq.~(\ref{con:Delta1s}), we have

\begin{subequations}  
\begin{numcases}{}
d\tilde{\Lambda}_2  =\frac{1}{2}\left[ \tilde{\Lambda}_1 (s) ,d\left[ \tilde{\Lambda}_1 (s)\right]\right] +\frac{1}{2}\left\{d\left[ \tilde{\Lambda}_1 (s)\right]\right\}^2  , \label{DL2}\\
d\tilde{\Lambda}_3 =\frac{1}{2}\left[ \tilde{\Lambda}_1 (s) ,d\left[ \tilde{\Lambda}_2 (s)\right]\right] + \frac{1}{2}\left[ \tilde{\Lambda}_2 (s) ,d\left[ \tilde{\Lambda}_1 (s)\right]\right] \notag\\
+ \frac{1}{6}\left[ \tilde{\Lambda}_1 (s),  \left[ \tilde{\Lambda} (s),d\left[ \tilde{\Lambda}_1 (s)\right]\right] \right]+ \frac{1}{6} \left[  \tilde{\Lambda}_1 (s) ,\left\{d\left[ \tilde{\Lambda}_1 (s)\right]\right\}^2 \right] \notag\\
+\frac{1}{12}\left[d\left[ \tilde{\Lambda}_1 (s)\right],\left[ \tilde{\Lambda}_1 (s) ,d\left[ \tilde{\Lambda}_1 (s)\right] \right]\right]. \label{DL3}
\end{numcases}
\end{subequations}
Similarly to the circumstance without measurements, the truncation error for the $j$th order Magnus expansion can be evaluated by $\left\| \tilde{\Lambda}_{j+1}(t) \right\|$. 
\begin{theorem} \label{Itoerror}
For the {It\^{o}} stochastic dynamics in Eq.~(\ref{con:SMEVect0}) with measurement at $t_0$, the truncation error for the averaged first-order Magnus expansion at time $t$ around $t_0$ is determined by the leading term $\mathbb{B}_{t_0}^2dt/2$.
\end{theorem}
\begin{proof}
For the last term in Eq.~(\ref{DL2}), according to Eq.~(\ref{con:dLambda1}), we have
\begin{equation}
\begin{aligned}
\frac{1}{2}\left\{d\left[ \tilde{\Lambda}_1 (t)\right]\right\}^2& = \frac{1}{2} \left[ \mathbb{A}\left(t\right)dt + \mathbb{B}_{t_0} dW_t \right]^2\\
& \approx  \frac{1}{2}\mathbb{B}_{t_0}^2dt+ \frac{1}{2} \left( \mathbb{A}\left(t\right)\mathbb{B}_{t_0} + \mathbb{B}_{t_0}\mathbb{A}\left(t\right) \right) dtdW_t\\
& \approx  \frac{1}{2}\mathbb{B}_{t_0}^2dt+ \frac{1}{2} \left\{ \mathbb{A}\left(t\right),\mathbb{B}_{t_0} \right\}dtdW_t,
\end{aligned}
\end{equation}
where $\left\{ \mathbb{A},\mathbb{B} \right\} =  \mathbb{A}\mathbb{B} + \mathbb{B}\mathbb{A}$ represents the anti-commutator, and the higher-order terms proportional to $\left(dt \right)^2$ are omitted. After averaging the stochastic term, the truncation error is determined by the leading term $\mathbb{B}_{t_0}^2dt/2$.
\end{proof}

\begin{remark}
Theorem~\ref{Itoerror} illustrates that the first-order Magnus expansion for the {\rm{It\^{o}}} stochastic Liouville equation can induce a truncation error proportional to $dt$, and the amplitude of this error component is determined by the measurement operator and the quantum state at the time of measurement. 
\end{remark}

Additionally, for the first component of Eq.~(\ref{DL2}), 
\begin{equation} \label{con:Itoerror1order}
\begin{aligned}
&\left[ \tilde{\Lambda}_1 (s) ,d\left[ \tilde{\Lambda}_1 (s)\right]\right] = \left[\int_{t_0}^{s} \mathbb{A}\left(\tau\right)d\tau +\mathbb{B}_{t_0} \Delta W_s,  \mathbb{A}\left(s\right)ds + \mathbb{B}_{t_0}  dW_s \right]  \\
= &  \int_{t_0}^{s} \left[ \mathbb{A}\left(\tau\right), \mathbb{A}\left(s\right) \right]d\tau ds + \int_{t_0}^{s} \left[ \mathbb{A}\left(\tau\right), \mathbb{B}_{t_0}   \right]d\tau dW_s +\left[ \mathbb{B}_{t_0} ,  \mathbb{A}\left(s\right)\right] \Delta W_s ds.
\end{aligned}
\end{equation}
According to Eq.~(\ref{con:Delta1s}), the deterministic term and incremental stochastic terms in the second line of Eq.~(\ref{con:Itoerror1order}) are higher-order infinitesimals than $dt$ when $s-t_0$ is of the same order as $dt$. 

For convenience, based on Eqs.~(\ref{con:Delta1s},\ref{con:Itoerror1order}), we rewrite $\tilde{\Lambda}_1(s)$ and $\tilde{\Lambda}_2(s)$ as
\begin{subequations}  \label{con:Delta12sRewrite}
\begin{numcases}{}
\tilde{\Lambda}_1(s)  =\tilde{\Lambda}_1^{(1)}(s) + \tilde{\Lambda}_1^{(2)}(s), \label{DL1RW}\\
\tilde{\Lambda}_2(s)  =\tilde{\Lambda}_2^{(1)}(s) + \tilde{\Lambda}_2^{(2)}(s) + \tilde{\Lambda}_2^{(3)}(s), \label{DL2RW}
\end{numcases}
\end{subequations}
where
\begin{subequations}  \label{con:Delta1sRewriteMeaning}
\begin{numcases}{}
\tilde{\Lambda}_1^{(1)}(s)  =\int_{t_0}^{s} \mathbb{A}\left(\tau\right)d\tau, \label{DL1RWMeaning1}\\
\tilde{\Lambda}_1^{(2)}(s)  =\mathbb{B}_{t_0} \Delta W_s, \label{DL1RWMeaning2}
\end{numcases}
\end{subequations}
and
\begin{subequations}  \label{con:Delta2sRewriteMeaning}
\begin{numcases}{}
\tilde{\Lambda}_2^{(1)}(s)  =\frac{1}{2}\int_{t_0}^{s} \left[\int_{t_0}^{s_1} \mathbb{A}\left(\tau\right)d\tau, \mathbb{B}_{t_0}   \right] dW_{s_1},\label{DL2RWMeaning1}\\
\tilde{\Lambda}_2^{(2)}(s)  =\frac{1}{2}\mathbb{B}_{t_0}^2\left(s-t_0\right) + \frac{1}{2}\int_{t_0}^{s} \left[ \int_{t_0}^{s_1}  \mathbb{A}\left(\tau\right)d\tau, \mathbb{A}\left(s_1\right) \right] ds_1 , \label{DL2RWMeaning2}\\
\tilde{\Lambda}_2^{(3)}(s)  =\frac{1}{2}\int_{t_0}^{s} \left\{ \mathbb{A}\left(s_1\right),\mathbb{B}_{t_0}\right\} dW_{s_1}ds_1 \notag \\
~~~~~~~~~~~~~+\frac{1}{2}\int_{t_0}^{s}\left[ \mathbb{B}_{t_0},  \mathbb{A}\left(s_1\right)\right] \Delta W_{s_1} ds_1. \label{DL2RWMeaning3}
\end{numcases} 
\end{subequations}

Based on this, we study the truncation errors for the stochastic evolution in $\left[t_0,t_0+h \right]$ in the following by taking $s =t_0+h$.

\subsection{Truncation error analysis for stochastic filtering dynamics}
According to Eq.~(\ref{con:PhisDef}), 
\begin{equation}
\begin{aligned}
\Phi_{t_0+h} & = \mathbb{I} + \tilde{\Lambda} \left(t_0+h \right) + \frac{\left[\tilde{\Lambda} \left(t_0+h \right) \right]^2}{2!} + \frac{\left[\tilde{\Lambda} \left(t_0+h \right) \right]^3}{3!}+ \cdots.
\end{aligned}
\end{equation}
Generalized from Eq.~(\ref{con:errordef}) for quantum dynamics without stochasticity, the truncation error of the $p$th order Magnus expansion for stochastic dynamics in $\left[t_0,t_0+h \right]$ can be evaluated as~\cite{krylov1980controlled}
\begin{equation} \label{con:Stochasticerrordef}
\begin{aligned}
\tilde{e}_p \left(h\right) &= \sqrt{\mathbb{E} \left\|\tilde{\Lambda} \left(t_0+h\right) -\sum_{j=1}^{p} \tilde{\Lambda}_j\left(t_0+h\right) \right\|^2},
\end{aligned}
\end{equation}
where $\mathbb{E}$ represents the mathematical expectation of a stochastic process. According to the Minkowski inequality,
\begin{equation} \label{con:MinkowskiInequal}
\begin{aligned} 
 \tilde{e}_p \left(h\right) \leq \sqrt{\sum_{j = p+1}^{\infty}\mathbb{E} \left\| \tilde{\Lambda}_j\left(t_0+h\right) \right\|^2} \leq C \sqrt{\mathbb{E} \left\| \tilde{\Lambda}_{p+1}\left(t_0+h\right) \right\|^2},
  \end{aligned}
\end{equation}
where $C$ is a positive constant.
\begin{theorem} \label{ItoStoMag}
The truncation error of the first-order Magnus expansion for the quantum  {It\^{o}} stochastic dynamics satisfies $\tilde{e}_1 \left(h\right) \leq C h $. 
\end{theorem}
\begin{proof}
For the three components in Eq.~(\ref{con:Delta2sRewriteMeaning}), $\tilde{\Lambda}_2^{(1)}(s)  $ and $\tilde{\Lambda}_2^{(3)}(s) $  are  of the order $ \left(s-t_0\right)^{3/2}$, $\tilde{\Lambda}_2^{(2)}(s) $ is of the leading order $\left(s-t_0\right)$.  The truncation error is dominated by the first component on the right-hand side of Eq.~(\ref{DL2RWMeaning2}).  Taking $s = t_0+h$, we can derive $\tilde{e}_1 \left(h\right) \leq C h $ according to Eq.~(\ref{con:MinkowskiInequal}), and the other terms of higher orders are omitted.
\end{proof}

\begin{remark}
Theorem~\ref{ItoStoMag} shows that the Magnus expansion based on quantum  {It\^{o}} stochastic dynamics can induce the term $\mathbb{B}_{t_0}^2\left(s-t_0\right)$ in Eq.~(\ref{DL2RWMeaning2}), rendering a larger truncation error proportional to $h$ compared to other terms  proportional to $h^{3/2}$. To eliminate the leading term due to quantum measurements, the second-order Magnus expansion is required.  
\end{remark}

\subsection{Stratonovich stochastic dynamics}
Generalized from the {\rm{It\^{o}}} stochastic master equation (\ref{con:SME}), the Stratonovich SME reads~\cite{moon2014interpretation}
\begin{equation} \label{con:StratonovichSME}
\begin{aligned}
d\rho_S   =&-i \left[H,\rho_S\right]dt  + \mathcal{L}_{O}\left[\rho_S\right]dt + \mathcal{H}[M]\rho_S\circ d W_t,
\end{aligned}
\end{equation}
where $\circ$ represents the Stratonovich product in stochastic equations. Similar to the {\rm{It\^{o}}} stochastic dynamics of $\vec{\rho}$ in Eq.~(\ref{con:SMEVec}), the Stratonovich dynamics in Eq.~(\ref{con:StratonovichSME}) can be vectorized as
\begin{equation} \label{con:SMEVecStratonovich}
\begin{aligned}
d\vec{\rho}_S   =&\mathbb{A}_S(t) \vec{\rho}_S dt  +\mathbb{B}_S\left(t, \rho_S\right)  \vec{\rho}_S \circ d W_t,
\end{aligned}
\end{equation}
where
\begin{subequations}  
\begin{numcases}{}
\mathbb{A}_S(t)  = \mathbb{A}\left(t\right)- \frac{1}{2}\mathbb{B}_{t_0}^2, \\
\mathbb{B}_S\left(t, \rho_S\right)  =\mathbb{B}\left(t_0, \rho_S\right),
\end{numcases} 
\end{subequations}
compared to Eq.~(\ref{con:SMEVect0}) for the small time step $\left[t_0,t_0+h \right]$.

Different from the {\rm{It\^{o}}}  stochastic dynamics in the subsection above, the Magnus expansion for Stratonovich stochastic dynamics can be derived similarly as Eq.~(\ref{con:LambdaDynamics}). This is due to the advantage of Stratonovich modeling in the Magnus expansion, and this has been systematically introduced in \cite{wang2020magnus}. Compared with the {\rm{It\^{o}}} stochastic dynamics, the advantage of the Stratonovich dynamics lies in the improvement of accuracy in the first-order Magnus expansion as below.  
\begin{theorem} \label{StrotovichStoMag}
The first-order Magnus expansion for the Stratonovich stochastic dynamics in Eq.~(\ref{con:SMEVecStratonovich}) satisfies $\tilde{e}_1 \left(h\right) \leq C h^{3/2} $. 
\end{theorem}
\begin{proof}
By replacing $\mathbb{A}\left(t_0\right)$ and $\mathbb{B}\left(t_0, \rho\left(t_0 \right)\right)$ in Eq.~(\ref{con:SMEVect0}) for the {\rm{It\^{o}}}  stochastic dynamics with $\mathbb{A}_S\left(t_0\right)$ and $\mathbb{B}_S\left(t_0, \rho_S\right) $ in Eq.~(\ref{con:Delta2sRewriteMeaning}) respectively, the component $(1/2)\mathbb{B}_{t_0}^2\left(s-t_0\right) $  in Eq.~(\ref{DL2RWMeaning2}) will be eliminated. Then the leading order of the truncation error is proportional to  $ \left(s-t_0\right)^{3/2}$. Taking $s = t_0+h$, we can derive $\tilde{e}_1 \left(h\right) \leq C h^{3/2} $, and the other higher-order terms are omitted.
\end{proof}

By combining Theorem~\ref{StrotovichStoMag} and the Magnus expansion based on {\rm{It\^{o}}} stochastic dynamics, the truncation errors between the {\rm{It\^{o}}} and Stratonovich modeling approaches lie mainly in the first-order Magnus expansion. The truncation errors of higher-order cases for both methods can be evaluated by the procedure outlined in \cite{wang2020magnus}.

\section{Numerical simulations} \label{Sec:numerical}
The numerical solutions and errors of the stochastic differential equation (\ref{con:SMEVec}) can be analyzed using Magnus expansion and Lie algebras, as in \cite{marjanovic2018numerical} for matrix differential equations, and in \cite{SIAM2003magnus,barfoot2026vector} for linear time-varying stochastic vector equations. The theoretical results in the former two sections are based on the non-Markovian Liouville equation, which is similar to the approach in \cite{SIAM2003magnus,barfoot2026vector}. In this section, we first take the simplest two-level system as an example, then generalize it to high-dimensional quantum systems.

\subsection{Two-level system}
For a two-level system, the density matrix can be represented as
\begin{equation} \label{con:twoleveldensitym}
\begin{aligned}
\rho = \begin{bmatrix}
\rho_{1} &\rho_{3}\\
\rho_{2} &\rho_{4}
\end{bmatrix},
\end{aligned}
\end{equation}
or in an equivalent vector format as $\vec{\rho}= \left[\rho_{1} ,\rho_{2} ,\rho_{3} ,\rho_{4} \right]^{\rm T}$~\cite{mironowicz2024semi,kunold2024vectorization}. Then both $\mathbb{A}(t)$ and $\mathbb{B}(t)$ are $4\times 4$ matrices, and the elements are determined by Eq.~(\ref{con:matrixAB}). We take  the Hamiltonian as
\begin{equation} \label{con:drhoct}
\begin{aligned}
H = \frac{\omega_0}{2} \sigma_z,
\end{aligned}
\end{equation}
where $\omega_0$ represents the resonant frequency of the atom, and $\sigma_z = \begin{pmatrix} 1 & 0 \\ 0 & -1 \end{pmatrix}$. 

\subsubsection{First-order Magnus expansion}
We take the system's coupling operator to the environment as 

\begin{equation}
\begin{aligned}
L = \sigma_- =  \begin{pmatrix} 0 & 0 \\ 1 & 0 \end{pmatrix},
\end{aligned}
\end{equation}
and $O = \int_0^t\alpha(t,s)ds L$.
\begin{figure}[h]
\centerline{\includegraphics[width=1\columnwidth]{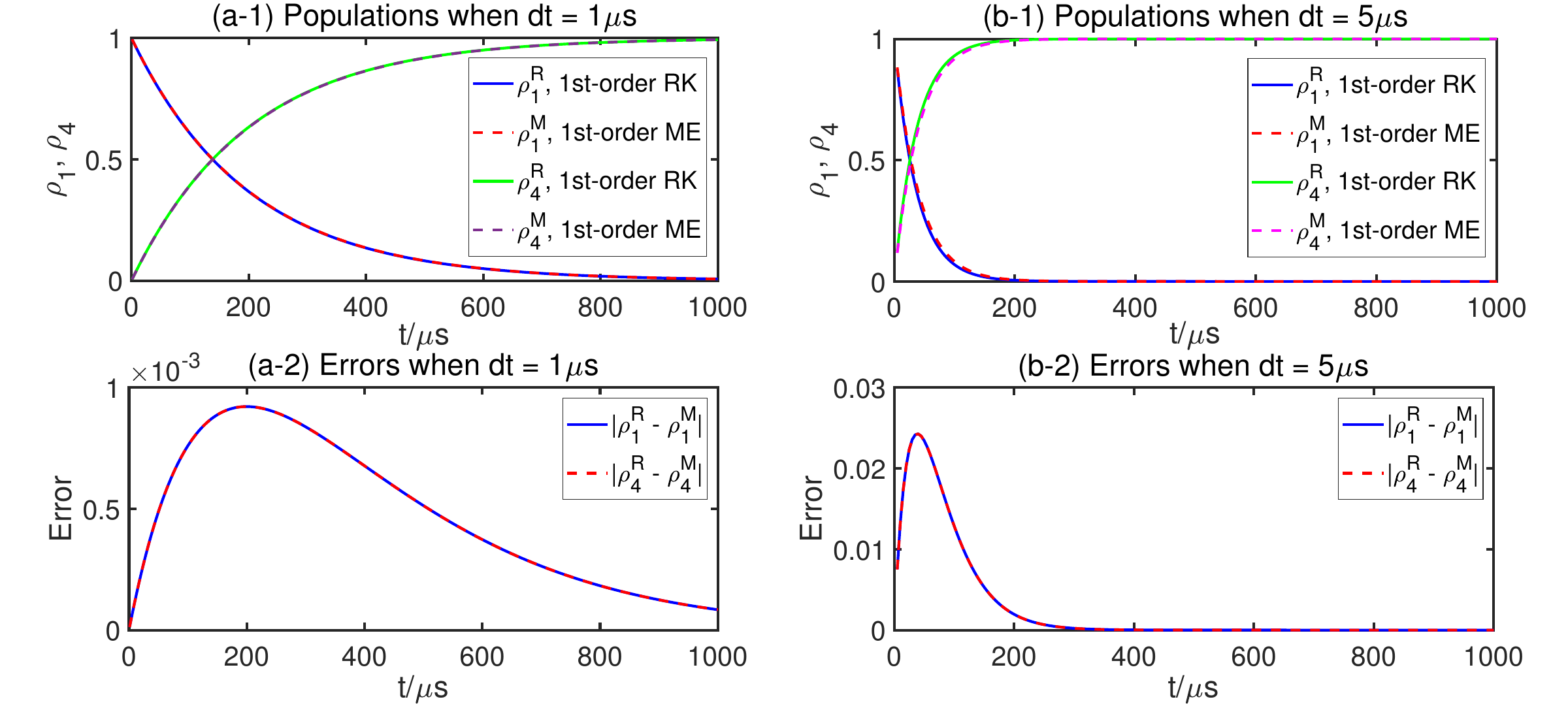}}
\caption{Comparisons between the first-order Runge-Kutta (RK) method and Magnus expansion (ME) on $\rho_{1}$ and $\rho_{4}$.}
	\label{fig:firstorder}
\end{figure}
In this simple setting, according to Eq.~(\ref{con:LambdaDynamics}), the higher-order Lie algebras in the Magnus expansion based on the vectorized Liouville equation all equal zero. In the following simulation, we take $\omega_0 = 50$GHz, $\gamma  = 50$KHz, $\Omega = 40$GHz for the integral kernel in Eq.~(\ref{con:NonMarkovAlpha}), and initially $\vec{\rho}(0)= \left[1 ,0,0 ,0 \right]^{\rm T}$. As simulated in Fig.~\ref{fig:firstorder}, the dynamical process represents the spontaneous emission of an excited two-level quantum system to the non-Markovian environment. The numerical simulation for $\rho_j$ with the Runge-Kutta method is denoted as $\rho_j^{\rm R}$, and that with the first-order Magnus expansion is denoted as $\rho_j^{\rm M}$. Because the Lie algebras in Eq.~(\ref{con:LambdaDynamics}) all equal zero, $\rho_j^{\rm M}$ is the precise solution of $\rho_j$. Then the comparisons in Fig.~\ref{fig:firstorder} show that the simulation error of the Runge-Kutta method increases with increasing $dt$.

\subsubsection{Second-order Magnus expansion} To clarify the numerical simulations with second-order Magnus expansions, we consider the operator as 
\begin{equation}  \label{con:Ltimevarying}
\begin{aligned}
L = \sigma_- + u(t)\sigma_x,
\end{aligned}
\end{equation}
where $\sigma_x = \begin{pmatrix} 0 & 1 \\ 1 & 0 \end{pmatrix}$, $u(t)$ represents the time-varying coupling between the quantum system and the environment through the operator $\sigma_x$, and $O = \int_0^t\alpha(t,s)ds L$. We take one Lindblad component as an example,
\begin{equation} \label{con:CalLind1}
\begin{aligned}
\left[L,\rho O^{\dag} \right] &=\int_0^t\alpha^*(t,s)ds \left[ \sigma_- + u(t)\sigma_x, \rho\left( \sigma_+ + u^*(t)\sigma_x\right) \right]\\
&=\int_0^t\alpha^*(t,s)ds  \left[ \sigma_- , \rho \sigma_+  \right] +u(t)\int_0^t\alpha^*(t,s)ds  \left[ \sigma_x , \rho \sigma_+  \right]\\
&~~~~+u^*(t)\int_0^t\alpha^*(t,s)ds  \left[ \sigma_- , \rho \sigma_x  \right] +\left|u(t)\right|^2\int_0^t\alpha^*(t,s)ds  \left[ \sigma_x , \rho \sigma_x  \right],
\end{aligned}
\end{equation}
where $\sigma_+ = \begin{pmatrix} 0 & 1 \\ 0 & 0 \end{pmatrix}$. For a general complex-valued $u(t)$, after vectorizing to the Liouville space, the corresponding items of Eq.~(\ref{con:CalLind1}) and its Hermite conjugation can be regarded as a special case of Eq.~(\ref{con:VtRewrite}) with $N=4$.

In the simulations in Fig.~\ref{fig:secondorder}, we take $u(t) = \sin(10t)$, (a-1)-(a-3) are for numerical simulations when $dt = 1\mu$s, (b-1)-(b-3) are for numerical simulations when $dt = 5\mu$s, and the other parameters are the same as those in Fig.~\ref{fig:firstorder}. Generalized from Fig.~\ref{fig:firstorder}, $\rho_j^{\rm M1}$ and $\rho_j^{\rm M2}$ in Fig.~\ref{fig:secondorder} represent the numerical simulations of $\rho_j$ via the first-order Magnus expansion and the second-order  Magnus expansion, respectively. 
\begin{figure}[h]
\centerline{\includegraphics[width=1\columnwidth]{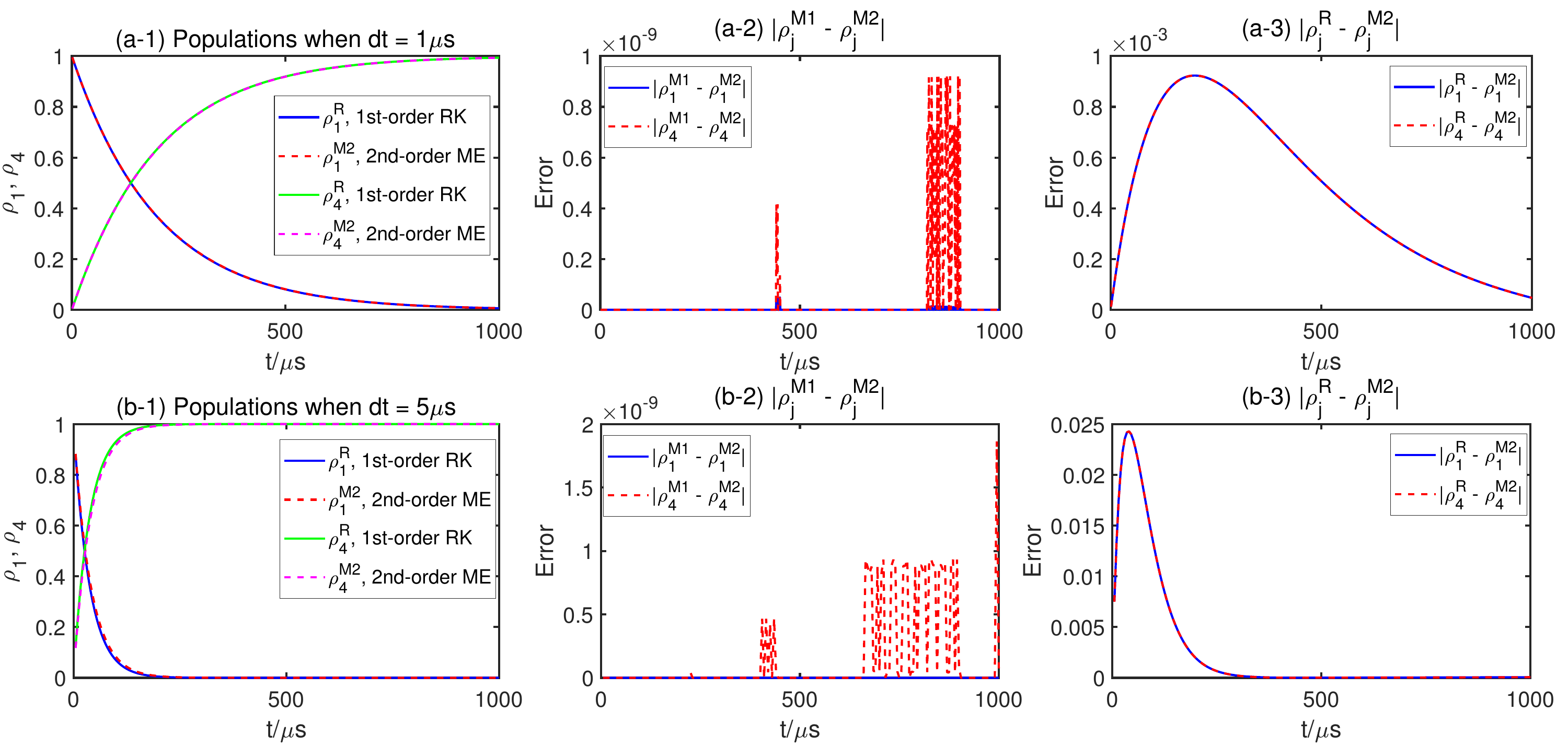}}
\caption{Comparisons among the first-order Runge-Kutta method, the first-order Magnus expansion and the second-order Magnus expansion on $\rho_{1}$ and $\rho_{4}$.}
	\label{fig:secondorder}
\end{figure}

Both simulations in Fig.~\ref{fig:firstorder} and Fig.~\ref{fig:secondorder} indicate that a larger $dt$ can induce faster convergence of populations in the spontaneous emission process. However, different from Fig.~\ref{fig:firstorder}, Fig.~\ref{fig:secondorder}(a-2) and Fig.~\ref{fig:secondorder}(b-2) further illustrate that the operator $L$ in Eq.~(\ref{con:Ltimevarying}) can generate higher-order Lie algebras in the second-order Magnus expansion. This can result in smaller truncation errors and higher accuracy compared to the first-order Runge-Kutta method, and the numerical results are further compared in Fig.~\ref{fig:secondorder}(a-3) and Fig.~\ref{fig:secondorder}(b-3).

\subsection{Error divergence with $h$} In this subsection, we compare the divergence of truncation errors based on the non-Markovian dynamics with the operator in Eq.~(\ref{con:Ltimevarying}). Taking the simulations with $dt = 10\mu$s, all the other parameters are the same as those in Fig.~\ref{fig:secondorder}. We further compare how the truncation error of the first-order Magnus expansion diverges in $\left[t_0,t_0+h \right]$ with $t_0 = 0$.  
\begin{figure}[h]
\centerline{\includegraphics[width=0.8\columnwidth]{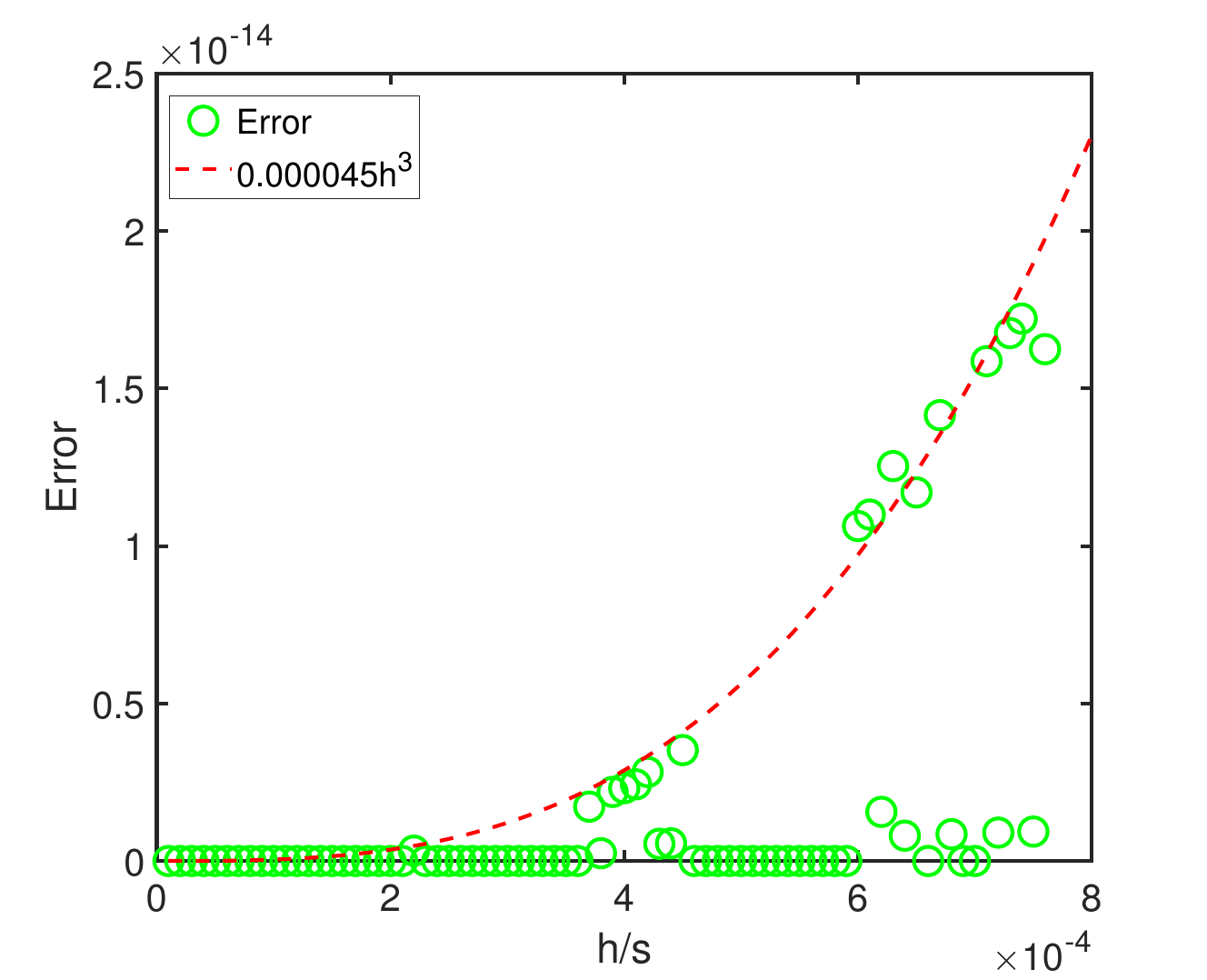}}
\caption{The truncation error of the first-order Magnus expansion influenced by $h$.}
	\label{fig:Compareerror}
\end{figure}

In Fig.~\ref{fig:Compareerror}, we take the dynamics of $\rho_1$ in Eq.~(\ref{con:twoleveldensitym}) as an example. For simplification, we regard the numerical simulation with the second-order Magnus expansion as precise solutions, then compare the divergence of numerical results with the first-order  Magnus expansion. As simulated by the green circles, the upper bound of truncation errors increases with $h$. By fitting with the function represented with the red line in Fig.~\ref{fig:Compareerror}, which is proportional to $h^3$, we can find that the numerical results agree with Eq.~(\ref{con:Omega2Middle}) and the simplified case in Theorem~\ref{FirstOrderError}.

\subsection{Two-atom filtering} According to the quantum filtering equation and numerical analysis in Sec.~\ref{Sec:Stochastic}, on the one hand, the filtering dynamics is influenced by the measurement time $t_0$ and the length of evolution time $h$.  Considering that the truncation error of the Magnus expansion in $\left[ t_0,t_0+h\right]$ has been given in  Sec.~\ref{Sec:Stochastic}, we clarify the numerical simulations in this section by taking $t_0 = 0$. 

On the other hand, the analysis on the convergence of the stochastic Magnus expansion requires that the measurement information can be acquired. However, in $\left[ t_0,t_0+h\right]$, only the measurement information at $t_0$ can be acquired according to Sec.~\ref{Sec:Stochastic} and physical realizations. In the following, to clarify the performance of numerical analysis, we first consider an ideal circumstance that the measurement results of quantum states $\rho(t)$ are available for arbitrary $t\in \left[t_0,t_0+h \right]$. Based on this, we clarify the performance of the Runge-Kutta and Magnus expansion methods in Eq.~(\ref{con:Delta12sRewrite}) only using the measurement information at $t_0$.

In the numerical simulations, we adopt the initial state as a two-body entangled state, and the density matrix can be represented as
\begin{equation}
\begin{aligned}
\rho\left(t_0\right) = 
 \begin{pmatrix} 
\rho_{\rm ee}\left(t_0 \right)  & 0 & 0 & \frac{1}{2} \\
 0 & \rho_{\rm eg}\left(t_0 \right) & 0 &0 \\
 0 & 0 & \rho_{\rm ge}\left(t_0 \right) &0\\
 \frac{1}{2} & 0 & 0 &\rho_{\rm gg}\left(t_0 \right)
 \end{pmatrix},
\end{aligned}
\end{equation}
with $\rho_{\rm ee}\left(t_0 \right) = \rho_{\rm gg}\left(t_0 \right) = 1/2$ and $\rho_{\rm eg}\left(t_0 \right) = \rho_{\rm ge}\left(t_0 \right) = 0$. According to the modeling in Eq.~(\ref{con:MasterEq}) and Eq.~(\ref{con:SME}), we denote $H =\left(\omega_0/2\right) \left(\sigma_1^z + \sigma_2^z\right) $,  $L = \sigma_1^- +  \sigma_2^- $, 
\begin{equation}  
\begin{aligned}
\sigma_1^z  =  \sigma_z  \otimes \begin{pmatrix} 1 & 0 \\ 0 & 1 \end{pmatrix},
~~~\sigma_2^z  = \begin{pmatrix} 1 & 0 \\ 0 & 1 \end{pmatrix} \otimes \sigma_z,
\end{aligned}
\end{equation}
\begin{equation}  
\begin{aligned}
\sigma_1^-  =  \sigma_-  \otimes \begin{pmatrix} 1 & 0 \\ 0 & 1 \end{pmatrix},
~~~\sigma_2^-  = \begin{pmatrix} 1 & 0 \\ 0 & 1 \end{pmatrix} \otimes \sigma_-,
\end{aligned}
\end{equation}
and $O = \int_0^t\alpha(t,s)ds L$. The quantum measurement is applied to the fist atom and the measurement operator reads $M = \sigma_1^z$.

\begin{figure}[h]
\centerline{\includegraphics[width=1\columnwidth]{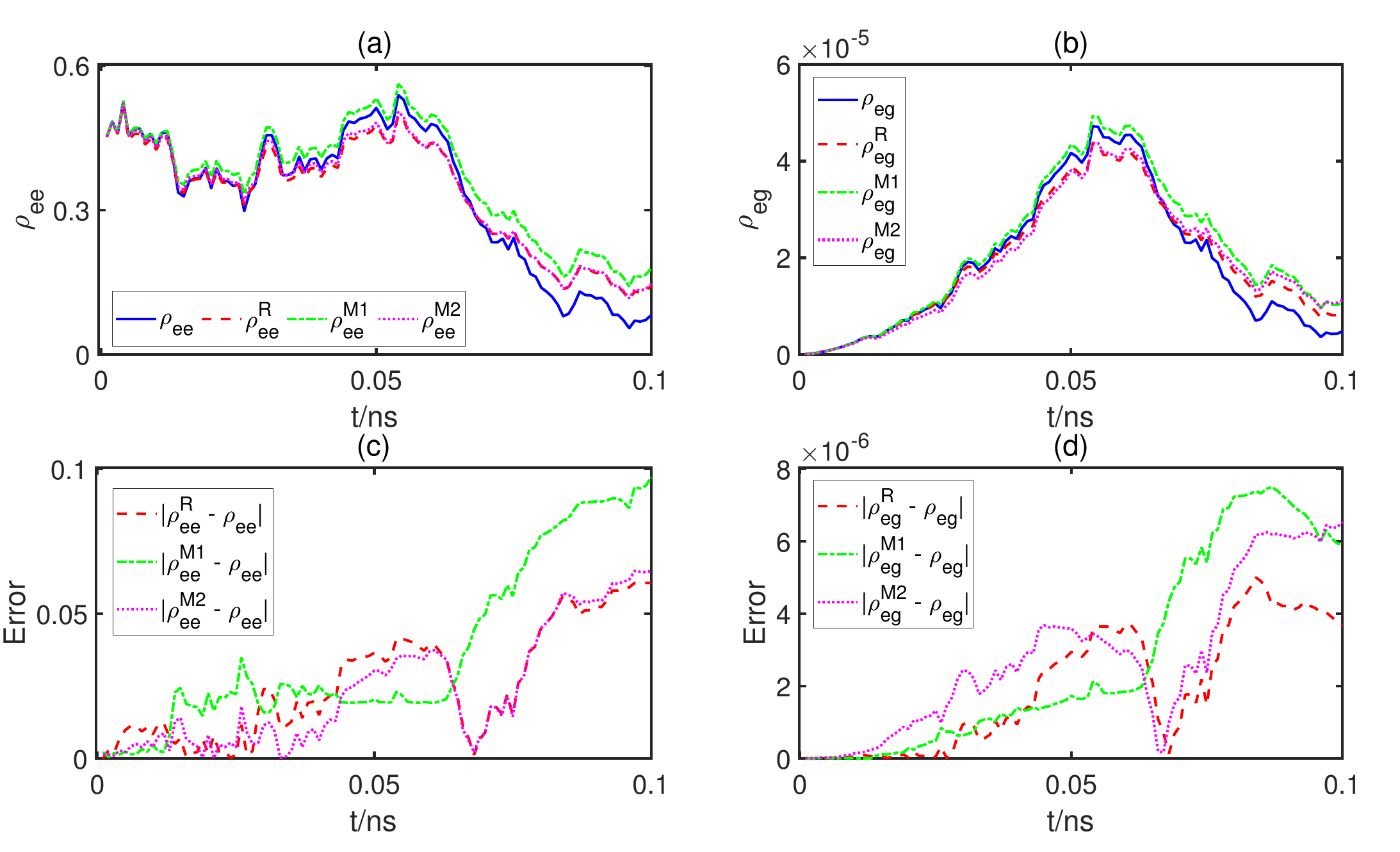}}
\caption{Numerical simulations on the quantum stochastic dynamics using the first-order Runge-Kutta method ($\rho_{\rm ee}^{\rm R}$, $\rho_{\rm eg}^{\rm R}$), the first-order Magnus expansion ($\rho_{\rm ee}^{\rm M1}$, $\rho_{\rm eg}^{\rm M1}$) and the second-order Magnus expansion ($\rho_{\rm ee}^{\rm M2}$, $\rho_{\rm eg}^{\rm M2}$).}
	\label{fig:Stochastic}
\end{figure}

In the numerical simulations in Fig.~\ref{fig:Stochastic}, we take the parameters as $\omega_0 = 50$GHz, $\Omega = 45$GHz, $\gamma  = 0.1$GHz, $dt = 0.001$ns and $h = 0.1$ns. The solid lines in Figs.~\ref{fig:Stochastic}(a) and (b) represent the ideal evolution trajectories simulated with the first-order Runge-Kutta method when the measurement results of quantum states $\rho(t)$ are available for arbitrary $t\in \left[t_0,t_0+h \right]$. The other lines in Figs.~\ref{fig:Stochastic}(a) and (b) represent the practical circumstance in Eq.~(\ref{con:SMEVect0}) in which only the measurement at $t_0$ is acquired. The dashed lines are simulated by the first-order Runge-Kutta method, and the other two lines are simulated by the first- and second-order Magnus expansions, respectively. The errors for Figs.~\ref{fig:Stochastic}(a) and (b) are further compared in Figs.~\ref{fig:Stochastic}(c) and (d), respectively.

The mathematical modeling of the ideal case with solid lines in Figs.~\ref{fig:Stochastic}(a) and (b) are different from the following RK and ME methods, because the real-time measurements rather than the initial measurement are adopted in the filtering. Here we take $h = 100dt$, and the errors between the ideal case and three numerical methods in  Figs.~\ref{fig:Stochastic}(c) and (d) indicate that the errors eventually tend to diverge as the simulation time increases. On the other hand, when $t<h/2$, the results are in agreement with each other and the errors are small in Figs.~\ref{fig:Stochastic}(c) and (d), clarifying that numerical methods with only the initial measurement at $t_0$ can simulate quantum stochastic dynamics in a short time scale.

\section{Conclusion} \label{Sec:Conclusion}
In this paper, we studied the non-Markovian interactions between the quantum system and environment from the perspective of Magnus expansion and Lie algebras. By modeling in the Liouville space, the original quantum master equation can be modeled as a high-dimensional time-varying equation with the time-varying components arising from the non-Markovian integrals. For the dynamics without measurement, the truncation errors in Magnus expansions are determined by commutators among the time-invariant Hamiltonian and time-varying components in the Liouville space arising from the non-Markovian environment. For stochastic dynamics with measurement, the first- and second-order Magnus expansions can be different according to whether the measurement noise is modeled in the {\rm{It\^{o}}}  or Stratonovich format. The numerical simulations further clarify the efficiency of the Magnus expansion in deterministic quantum evolutions and quantum stochastic dynamics with small step sizes.

\appendix
\section{Derivation on the non-Markovian model in the interaction picture} \label{Sec:IntractionPicture}
In this section, we introduce the derivation of the non-Markovian master equation in the interaction picture. 

The dynamics of the density matrix in the interaction picture can be derived as
\begin{equation} \label{con:dRhoI}
\begin{aligned}
\frac{d\rho_{\rm I}}{dt} &= \frac{dU_0^{\dag} }{dt}\rho U_0 + U_0^{\dag}\frac{d\rho}{dt}U_0 + U_0^{\dag}\rho\frac{dU_0}{dt}\\
&=iH_0U_0^{\dag}\rho U_0 + U_0^{\dag}\frac{d\rho}{dt}U_0 -iU_0^{\dag}\rho H_0U_0\\
& = i\left[H_0, U_0^{\dag}\rho U_0\right]+ U_0^{\dag}\frac{d\rho}{dt}U_0.
\end{aligned}
\end{equation}
On the other hand, we consider the unitary rotation upon Eq.~(\ref{con:MasterEq}), namely
\begin{equation}
\begin{aligned}
U_0^{\dag} \left[H,\rho\right]U_0 &= U_0^{\dag} \left(H\rho-\rho H \right) U_0\\
& = U_0^{\dag} H U_0   U_0^{\dag} \rho  U_0  - U_0^{\dag} \rho U_0 U_0^{\dag} H U_0 \\
& = \left[ H_0 + H_{\rm CI},\rho_{\rm I} \right],
\end{aligned}
\end{equation}
\begin{equation}
\begin{aligned}
U_0^{\dag} \left[L,\rho O^{\dag} \right] U_0 &= U_0^{\dag} \left(L\rho O^{\dag}-\rho O^{\dag} L \right) U_0\\
& = U_0^{\dag} L U_0 U_0^{\dag} \rho U_0 U_0^{\dag} O^{\dag}  U_0 - U_0^{\dag} \rho  U_0 U_0^{\dag} O^{\dag}  U_0 U_0^{\dag} L U_0\\
& = L_{\rm I} \rho_{\rm I}  O_{\rm I}^{\dag} -   \rho_{\rm I}  O_{\rm I}^{\dag}   L_{\rm I},
\end{aligned}
\end{equation}
and
\begin{equation}
\begin{aligned}
U_0^{\dag} \left[O\rho, L^{\dag}\right]  U_0 &= U_0^{\dag} \left(O\rho L^{\dag}-L^{\dag} O\rho \right) U_0\\
& = U_0^{\dag} O U_0 U_0^{\dag} \rho  U_0 U_0^{\dag} L^{\dag} U_0 - U_0^{\dag} L^{\dag} U_0 U_0^{\dag} O U_0 U_0^{\dag} \rho U_0\\
& = O_{\rm I} \rho_{\rm I} L_{\rm I}^{\dag} - L_{\rm I}^{\dag} O_{\rm I} \rho_{\rm I}.
\end{aligned}
\end{equation}
Then combined with Eq.~(\ref{con:dRhoI}), we have
\begin{equation} \label{con:dRhoIV2}
\begin{aligned}
\frac{d\rho_{\rm I}}{dt}& = i\left[H_0, U_0^{\dag}\rho U_0\right] -i  \left[ H_0 + H_{\rm CI} ,\rho_{\rm I} \right] + \left[ L_{\rm I}, \rho_{\rm I}  O_{\rm I}^{\dag}  \right] + \left[ O_{\rm I} \rho_{\rm I}, L_{\rm I}^{\dag}\right]\\
&  =-i \left[ H_{\rm CI} ,\rho_{\rm I} \right]+ \left[ L_{\rm I}, \rho_{\rm I}  O_{\rm I}^{\dag}  \right] + \left[ O_{\rm I} \rho_{\rm I}, L_{\rm I}^{\dag}\right].
\end{aligned}
\end{equation}

\bibliographystyle{siamplain}
\bibliography{bib}
\end{document}